\documentclass[11pt]{article}
\usepackage[a4paper,margin=25mm]{geometry}
\usepackage{amsmath,amssymb,amsthm,mathtools,bm}
\usepackage{booktabs}
\usepackage{longtable}
\usepackage{array}
\usepackage{multirow}
\usepackage{etoolbox}
\usepackage{pgfplots}
\usepackage[colorlinks=true,linkcolor=blue,urlcolor=blue,citecolor=blue]{hyperref}
\usepackage{graphicx}
\usepackage{adjustbox}
\pgfplotsset{compat=1.18}

\newcommand{\Ftwo}{\mathbb{F}_2}
\newcommand{\Z}{\mathbb{Z}}

\newcommand{\Ker}{\mathrm{Ker}}
\newcommand{\Row}{\mathrm{Row}}
\newcommand{\rank}{\mathrm{rank}}
\newcommand{\wt}{\mathrm{wt}}
\newcommand{\affexpr}[2]{\ensuremath{#1x+#2}}
\newcommand{\diag}{\mathrm{diag}}
\newcommand{\Hhat}{\hat{H}}
\newcommand{\Htil}{\tilde{H}}

\newcommand{\dbound}[1]{\ensuremath{\bar{d}^{\mathrm{#1}}}}
\newcommand{\dlat}{\dbound{lat}}
\newcommand{\drls}{\dbound{rls}}
\newcommand{\dblk}{\dbound{blk}}
\newcommand{\dfib}{\dbound{fib}}
\newcommand{\dcrt}{\dbound{crt}}
\newcommand{\ddir}{\dbound{dir}}
\newcommand{\dets}{\dbound{ets}}
\newcommand{\ddec}{\dbound{dec}}
\newcommand{\ubhit}[1]{\begingroup\setlength{\fboxsep}{1pt}\colorbox{gray!18}{\strut #1}\endgroup}
\newcommand{\markededit}[1]{#1}
\newcommand{\vect}[1]{\bm{#1}}
\newcommand{\vx}{\vect{x}}
\newcommand{\vy}{\vect{y}}
\newcommand{\vz}{\vect{z}}
\newcommand{\vu}{\vect{u}}
\newcommand{\vvct}{\vect{v}}

\newcommand{\ve}{\vect{e}}
\newcommand{\vlambda}{\vect{\lambda}}

\newtheoremstyle{uprightthm}
  {}{}{\normalfont\upshape}{}{\bfseries}{.}{ }{}
\theoremstyle{uprightthm}
\newtheorem{theorem}{Theorem}[section]
\newtheorem{lemma}[theorem]{Lemma}
\newtheorem{definition}[theorem]{Definition}
\newtheorem{proposition}[theorem]{Proposition}

\newtheorem{corollary}[theorem]{Corollary}
\newtheorem{example}[theorem]{Example}

\AtBeginEnvironment{theorem}{\pushQED{\qed}}
\AtEndEnvironment{theorem}{\popQED}
\AtBeginEnvironment{lemma}{\pushQED{\qed}}
\AtEndEnvironment{lemma}{\popQED}
\AtBeginEnvironment{definition}{\pushQED{\qed}}
\AtEndEnvironment{definition}{\popQED}
\AtBeginEnvironment{proposition}{\pushQED{\qed}}
\AtEndEnvironment{proposition}{\popQED}
\AtBeginEnvironment{remark}{\pushQED{\qed}}
\AtEndEnvironment{remark}{\popQED}
\AtBeginEnvironment{corollary}{\pushQED{\qed}}
\AtEndEnvironment{corollary}{\popQED}
\AtBeginEnvironment{example}{\pushQED{\qed}}
\AtEndEnvironment{example}{\popQED}

\title{Heuristic Search for Minimum-Distance Upper-Bound Witnesses in Quantum APM-LDPC Codes}
\author{Kenta Kasai\\
Institute of Science Tokyo\\
\texttt{kenta@ict.eng.isct.ac.jp}}
\date{}

\begin{document}
\maketitle

\begin{abstract}
This paper investigates certified upper bounds on the minimum distance of an
explicit family of Calderbank-Shor-Steane quantum LDPC codes constructed from
affine permutation matrices. All codes considered here have active Tanner graphs
of girth eight. Rather than attempting to prove a general lower bound for the
full code distance, we focus on constructing low-weight non-stabilizer logical
representatives, which yield valid upper bounds once they are verified to lie in
the opposite parity-check kernel and outside the stabilizer row space. We develop
a unified framework for such witnesses arising from latent row relations,
restricted-lift subspaces including block-compressed, selected-fiber, and
CRT-stripe constructions, cycle-8 elementary trapping-set structures, and
decoder-failure residuals. In every case, search is used only to generate
candidates; the reported bounds begin only after explicit kernel and row-space
exclusion tests have been passed. For the latent part, we also identify a
block-compression criterion under which the certification becomes exact.
Applying these methods to representative APM-LDPC codes sharpens previously
reported upper bounds and provides concrete certified values across the explored
parameter range.
\end{abstract}

\section{Introduction}
Low-density parity-check (LDPC) codes use sparse parity-check matrices, or
equivalently sparse Tanner graphs, to support efficient decoding while retaining
nontrivial minimum distance~\cite{Gallager1963,Tanner1981}. In the classical
setting, random regular ensembles, expander/Tanner constructions,
quasi-cyclic constructions, and protograph liftings give different ways to
combine sparse local constraints with global
distance~\cite{SipserSpielman1996,Fossorier2004,MitchellSmarandacheCostello2014}.
In quantum coding theory, these classical constructions are directly relevant through
Calderbank--Shor--Steane (CSS) codes, which connect quantum codes to pairs of
classical linear codes satisfying an orthogonality relation~\cite{CalderbankShor1996,Steane1996}.

For quantum LDPC codes, however, the classical design principles do not apply
without additional constraints. One needs two sparse parity-check matrices that
are not only individually useful but also mutually orthogonal in the CSS sense.
This compatibility condition appears in early sparse-graph quantum code
constructions~\cite{MacKay2004,HagiwaraImai2007}, in hypergraph-product codes
with positive rate and square-root distance~\cite{TillichZemor2014}, and in
modern constructions of asymptotically good quantum LDPC codes and quantum
Tanner codes~\cite{PanteleevKalachev2022,LeverrierZemor2022,DinurHsiehLinVidick2023}. Much of that work
gives lower bounds or asymptotic existence results.
Recent related work also studies near-bound finite-length constructions,
random quantum-LDPC constructions, and finite-degree
Gilbert--Varshamov-type constructions~\cite{KomotoKasai2025,OkadaKasai2025,KasaiGV2026}.

For minimum-distance evaluation, a lower bound for the full code distance would
be the strongest conclusion, because it \markededit{gives}
guaranteed error-correction capability rather than merely excluding larger
values. In the present family, however, the known lower-bound techniques do not
transfer directly to the full distance, and the classical literature on exact
minimum-weight search is primarily effective at short and medium blocklengths.
Even the elementary tree-bound baseline for sparse Tanner graphs~\cite{Tanner1981}
gives only a small constant lower-bound baseline for the active girth-eight graphs used below; it
does not \markededit{explain} the larger distance values observed in the computations.
This paper therefore records explicit upper-bound vectors and verifies them
inside one specified construction family. More detailed comments on the lower-bound literature and on the
classical short-code literature are deferred to later sections.

Affine permutation matrix (APM) constructions impose the CSS orthogonality
condition in a controlled way~\cite{MyungYangPark2006,KasaiEtAl2011,KomotoKasai2025,KasaiGirth16APM2025,Kasai2026}.
Instead of imposing commutativity on the full parent matrix, one enforces
orthogonality only on the active part that is actually used as the CSS
stabilizer, while allowing controlled non-commutativity in the latent part.
This produces high-girth binary CSS LDPC codes together with explicit latent
row spaces that may contain low-weight logical representatives.  The closest antecedent is~\cite{Kasai2026}, which developed the
same structured family primarily from the construction and decoding side. The
present paper is a follow-up on that same construction method, focused on how
tightly one can upper-bound and update the minimum distance within the family.
Recent nonbinary and affine-permutation studies of quantum-LDPC decoding and
high-girth design provide complementary finite-length evidence for the same
design direction~\cite{KasaiISIT2025,KasaiDegeneracy2025,KomotoKasaiSharpBP2025,KasaiGirth16APM2025}.
Following that terminology, we refer to this family as quantum APM-based LDPC
(APM-LDPC) codes.
A standing hypothesis in this paper is that every concrete code analyzed or
reported below has active Tanner graphs of girth 8. This girth-eight condition
is part of the search constraint, not only a regularity assumption.  It
excludes four- and six-cycles, so the first cycle-based local
supports used by the elementary trapping set (ETS) method are built from
8-cycles.

The choice of lift sizes in the present search is not arbitrary. The companion
manuscript on affine commutation patterns~\cite{KasaiCommPattern2026} shows
that the cross-commuting nonabelian affine pattern needed on the construction
side is not available for prime-power moduli and is available for
Chinese remainder theorem (CRT)-split moduli. For that reason, the
computational search here is organized around non-prime-power lift sizes.

This leads to a limited question. The family proposed
in~\cite{Kasai2026} may have distance that grows with blocklength.  The present
paper does not try to decide this.  Instead, it asks which low-weight logical
operators are found by the searches used here and then \markededit{pass} the stated
kernel and row-space tests.  The upper bound reported in~\cite{Kasai2026} is updated, while the
overall asymptotic behavior of the family remains unsettled.

The paper is therefore deliberately split into a proof part and a
computational verification part.  The mathematical part states sufficient conditions under which
a vector in a specified subspace is a non-stabilizer logical representative.
The computational verification part supplies vectors found by searches and checks
those conditions for representative codes.  Search procedures are used only to
find candidate vectors; the upper-bound claims in the paper begin only after the
appropriate parity-check-kernel and row-space-exclusion tests have been \markededit{passed}.
A compact classification of the upper-bound methods is included at the end of
Section~\ref{sec:framework-en}.

\section{Parent, Active, and Latent Matrices}
\label{sec:framework-en}
Let
\[
\Hhat_X\in\Ftwo^{r_X\times n},\qquad
\Hhat_Z\in\Ftwo^{r_Z\times n}
\]
be binary parent matrices with the same number of columns.
Choose subsets of rows and call them the active rows.
The complementary rows are called latent.
After permuting rows if necessary, we may write
\[
\Hhat_X=\begin{bmatrix}H_X\\ \Htil_X\end{bmatrix},\qquad
\Hhat_Z=\begin{bmatrix}H_Z\\ \Htil_Z\end{bmatrix},
\]
where $H_X,H_Z$ are the active matrices and $\Htil_X,\Htil_Z$ are the latent matrices.
This section makes no APM or block-size assumptions; it only defines the
overall, latent, and non-latent distances attached to this decomposition.

Vectors are regarded as elements of the ambient binary vector space, and their
orientation is chosen only when a matrix product is written.  For
$A\in\Ftwo^{r\times n}$, we write
\[
\Ker(A):=\{\vu\in\Ftwo^n:A\vu=0\}.
\]
We also write $\Row(A)\subseteq\Ftwo^n$ for the subspace spanned by the rows of
$A$.  For a binary vector $\vu$, $\wt(\vu)$ denotes its Hamming weight, namely
the number of nonzero coordinates.  When a row-space element is displayed as
$\vx=\vlambda^{\mathsf T}A$,
$\vlambda$ is a column coefficient vector and $\vx$ is the corresponding row
combination.  Syndrome equations are then written using the transpose, for
example $B\vx^{\mathsf T}=0$; equivalently, a displayed row vector $\vx$
belongs to $\Ker(B)$ precisely when $B\vx^{\mathsf T}=0$.  This convention is
used throughout to avoid introducing separate symbols for the same binary vector
in row and column form.

\begin{definition}[Overall, latent, and non-latent distances]
\label{def:latent-nonlatent-distances-en}
From this point on, assume that $H_XH_Z^{\mathsf T}=0$ and define the CSS constituent codes by
\[
C_X:=\Ker(H_X),\qquad C_Z:=\Ker(H_Z).
\]
The usual CSS distances are
\[
d_X=\min\{\wt(\vx):\vx\in C_Z\setminus C_X^\perp\},\qquad
d_Z=\min\{\wt(\vz):\vz\in C_X\setminus C_Z^\perp\}.
\]
To isolate the contribution coming directly from the latent row spaces, define
\[
d_X^{(\mathrm{lat})}:=
\min\{\wt(\vx):\vx\in (C_Z\cap \Row(\Htil_X))\setminus C_X^\perp\},
\]
\[
d_Z^{(\mathrm{lat})}:=
\min\{\wt(\vz):\vz\in (C_X\cap \Row(\Htil_Z))\setminus C_Z^\perp\}.
\]
To measure the complementary contribution outside the latent row spaces, define
\[
d_X^{(\mathrm{nlat})}:=
\min\{\wt(\vx):\vx\in C_Z\setminus C_X^\perp,\ \vx\notin \Row(\Htil_X)\},
\]
\[
d_Z^{(\mathrm{nlat})}:=
\min\{\wt(\vz):\vz\in C_X\setminus C_Z^\perp,\ \vz\notin \Row(\Htil_Z)\}.
\]
All four minima are interpreted as $+\infty$ if the corresponding set is empty.
Then
\[
d_X=\min\{d_X^{(\mathrm{lat})},d_X^{(\mathrm{nlat})}\},
\qquad
d_Z=\min\{d_Z^{(\mathrm{lat})},d_Z^{(\mathrm{nlat})}\}.
\]
\end{definition}

This definition separates logical representatives contained in the latent row
spaces from logical representatives outside those row spaces.  Later sections
state upper-bound criteria for both classes.

Before specializing to APMs, Table~\ref{tab:upper-bound-methods-en} fixes the
notation for the upper-bound methods used later.  A reported upper-bound
witness means an explicitly given vector that satisfies the relevant
parity-check-kernel condition and is outside the corresponding stabilizer row
space.  The latent and restricted-lift methods specify subspaces in which such
vectors are sought.  The direct, cycle-8 ETS, and decoder-failure methods
specify other ways to generate candidate vectors or supports.  In every case,
the upper-bound claim is the same: the final vector must pass the stated kernel
and row-space tests.

\begin{table}[t]
\centering
\caption{Upper-bound methods.}
\label{tab:upper-bound-methods-en}
\small
\setlength{\tabcolsep}{5pt}
\renewcommand{\arraystretch}{1.12}
\begin{tabular}{p{0.38\linewidth}p{0.17\linewidth}p{0.30\linewidth}}
\toprule
Name & Notation & Statement \\
\midrule
latent upper bound & \dlat & Proposition~\ref{prop:latent-upper-en} \\
full-fiber block-compression upper bound & \dblk & Proposition~\ref{prop:nonlatent-upper-en} \\
\markededit{proper selected-fiber restricted-lift upper bound} & \dfib & Proposition~\ref{prop:fiber-quotient-upper-en} \\
CRT-stripe restricted-lift upper bound & \dcrt & Proposition~\ref{prop:crt-compression-upper-en} \\
direct CSS-search upper bound & \ddir & Proposition~\ref{prop:direct-css-upper-en} \\
cycle-8 ETS upper bound, with stage & \dets & Corollary~\ref{cor:cycle8-css-upper-en} \\
decoder-failure upper bound & \ddec & Proposition~\ref{prop:decoder-failure-upper-en} \\
\bottomrule
\end{tabular}
\end{table}

The notation in Table~\ref{tab:upper-bound-methods-en} records displayed
weights of \markededit{vectors found by the searches and checked by the stated tests}; it is not a new set of
intrinsic distance parameters and it is not the result of an exhaustive
minimization over all possible witnesses.  For
\[
\star\in\{\mathrm{lat},\mathrm{blk},\mathrm{fib},\mathrm{crt},
\mathrm{dir},\mathrm{ets},\mathrm{dec}\},
\]
an entry
\[
\bar d_X^{\star}=w
\qquad
\text{or}
\qquad
\bar d_Z^{\star}=w
\]
means that \markededit{the corresponding method has produced an $X$- or
$Z$-type logical representative of weight $w$}.  If \markededit{no vector of the specified
side has been reported}, the corresponding value is displayed as ``--'' in the
tables.  \markededit{Thus each finite entry is an upper bound on the
corresponding CSS distance because it is backed by an explicit witness satisfying
the stated tests, but it is not claimed to be the smallest value obtainable by
that method.}  A later search may add a smaller entry.
The notation \drls is only an aggregate label for the three restricted-lift
methods \dblk, \dfib, and \dcrt; the tables display these three component
entries separately.
\markededit{In the tables, \dblk records the full-fiber case \(S=\Z/m\Z\), while \dfib
records proper selected-fiber patterns \(\emptyset\neq S\subsetneq\Z/m\Z\).
Thus the displayed \dfib column is not minimized over the full-fiber case
already reported in \dblk.  If one merged the two columns into a single
fiber-pattern value, the merged entry would be at most \dblk, but the separated
table intentionally keeps the full-fiber and proper-fiber witnesses distinct.}

\section{APM Specialization}
\label{sec:apm-specialization-en}
We now specialize the general parent/active/latent framework to the APM family used throughout the computational part.
Let $L$ be even and let $F_0,\dots,F_{L/2-1}$ and $G_0,\dots,G_{L/2-1}$ be $P\times P$ permutation matrices. We define the parent matrices of the APM-LDPC family by
\[
(\Hhat_X)_{i,j}=F_{j-i},\qquad
(\Hhat_X)_{i,L/2+j}=G_{j-i},
\]
\[
(\Hhat_Z)_{i,j}=G_{i-j}^{\mathsf T},\qquad
(\Hhat_Z)_{i,L/2+j}=F_{i-j}^{\mathsf T},
\]
with all indices interpreted modulo $L/2$.
For the active choice used in this paper, the first $J$ block rows are taken as active and the remaining block rows are latent, in the sense of Section~\ref{sec:framework-en}.

For each $r\in \Z/(L/2)\Z$, define
\[
\Psi_r:=\sum_{u=0}^{L/2-1}\bigl(F_uG_{r-u}+G_{r-u}F_u\bigr)
\]
over $\Ftwo$.
The blocks of $\Hhat_X\Hhat_Z^{\mathsf T}$ depend only on the row difference and are given by the corresponding $\Psi_r$.

For the standard active choice, let
\[
\Delta:=\{(k-i)\bmod(L/2):0\le i,k\le J-1\}.
\]
This sufficient criterion is already present in the $(3,12)$ setting of~\cite{Kasai2026}; here we state it in the general $(J,L)$ form.
Indeed, the $(i,k)$ block of $H_XH_Z^{\mathsf T}$ is $\Psi_{k-i}$ with indices taken modulo $L/2$.
Since $0\le i,k\le J-1$, the difference $k-i$ always belongs to $\Delta$.
Thus the active orthogonality condition is guaranteed by
\[
\Psi_r=0\qquad (r\in\Delta).
\]
Thus CSS orthogonality is equivalent to the vanishing of the mixed products
$\Psi_r$ for the active residues.

In the concrete family treated computationally in this paper we fix
\[
J=3,\qquad L=12,\qquad L/2=6.
\]
Then
\[
\Delta=\{0,1,2,4,5\},
\]
so the only unconstrained residue is $r=3$. The search problem is therefore to realize
\[
\Psi_r=0\quad (r\neq 3),\qquad \Psi_3\neq 0,
\]
so that the only nonzero active--latent mixed product allowed by the template is
$\Psi_3$.

Each block is an affine permutation
\[
x\mapsto ax+b\pmod P,\qquad \gcd(a,P)=1.
\]
This representation is used for two reasons. First, commutativity reduces to a
linear congruence. Second, if $P=mQ$, the same affine rule induces a
well-defined permutation modulo $Q$, which is the algebraic fact used in block
compression.

\section{Latent Upper Bounds}
\label{sec:latent-en}
This section evaluates the latent distances for a general regular APM-LDPC
construction.  Let the active row blocks be indexed by $0,\dots,J-1$ and the
latent row blocks by $J,\dots,L/2-1$, and assume $H_XH_Z^{\mathsf T}=0$.

\begin{lemma}[General active--latent mixed product]
\label{lem:latent-mixed-product-en}
For every active index $0\le i\le J-1$ and latent index $0\le \ell\le L/2-J-1$,
\[
[H_Z\Htil_X^{\mathsf T}]_{i,\ell}
=\Psi_{(J+\ell-i)\bmod (L/2)}^{\mathsf T},
\qquad
[H_X\Htil_Z^{\mathsf T}]_{i,\ell}
=\Psi_{(J+\ell-i)\bmod (L/2)}.
\]
\end{lemma}

\begin{proof}
The $(i,\ell)$ block of $H_Z\Htil_X^{\mathsf T}$ is the parent mixed product between active row block $i$ and latent row block $J+\ell$, which is exactly
\[
\Psi_{(J+\ell-i)\bmod(L/2)}^{\mathsf T}.
\]
The $H_X\Htil_Z^{\mathsf T}$ formula is analogous.
\end{proof}

The lemma reduces the latent upper-bound problem from the full parent matrix to the kernel of an explicit mixed product. The next proposition states the resulting latent upper-bound criterion, in the same spirit as the latent-distance analysis developed for this APM framework in earlier work~\cite{Kasai2026}. Operationally, finding one feasible column coefficient vector $\vlambda$ is enough: the row combination $\vlambda^{\mathsf T}\Htil_X$ or $\vlambda^{\mathsf T}\Htil_Z$ is then an explicit logical representative, and its weight gives the upper bound.

\begin{proposition}[Latent upper bounds from mixed-product kernels]
\label{prop:latent-upper-en}
The latent distances admit the exact parameterizations
\[
d_X^{(\mathrm{lat})}
=
\min\left\{
\wt(\vlambda^{\mathsf T}\Htil_X):
\vlambda\in(\Ftwo^P)^{L/2-J},\ 
H_Z\Htil_X^{\mathsf T}\vlambda=0,\ 
\vlambda^{\mathsf T}\Htil_X\notin C_X^\perp
\right\},
\]
\[
d_Z^{(\mathrm{lat})}
=
\min\left\{
\wt(\vvct^{\mathsf T}\Htil_Z):
\vvct\in(\Ftwo^P)^{L/2-J},\ 
H_X\Htil_Z^{\mathsf T}\vvct=0,\ 
\vvct^{\mathsf T}\Htil_Z\notin C_Z^\perp
\right\},
\]
with the convention that the minimum is $+\infty$ when the corresponding set is empty.

In particular, every feasible coefficient vector
$\vlambda\in(\Ftwo^P)^{L/2-J}$ satisfying
$H_Z\Htil_X^{\mathsf T}\vlambda=0$ and
$\vlambda^{\mathsf T}\Htil_X\notin C_X^\perp$
produces
$\vx:=\vlambda^{\mathsf T}\Htil_X\in (C_Z\cap \Row(\Htil_X))\setminus C_X^\perp$,
hence $d_X^{(\mathrm{lat})}\le \wt(\vx)$ and therefore $d_X\le \wt(\vx)$.
The $Z$-side statement is identical.
\end{proposition}

\begin{proof}
We prove the $X$-side identity. By Definition~\ref{def:latent-nonlatent-distances-en},
\[
d_X^{(\mathrm{lat})}
=
\min\{\wt(\vx):\vx\in (C_Z\cap \Row(\Htil_X))\setminus C_X^\perp\}.
\]
Let
\[
\mathcal{L}_X
:=
\left\{
\vlambda\in(\Ftwo^P)^{L/2-J}:
H_Z\Htil_X^{\mathsf T}\vlambda=0,\ 
\vlambda^{\mathsf T}\Htil_X\notin C_X^\perp
\right\}.
\]

If $\vlambda\in\mathcal{L}_X$ and $\vx:=\vlambda^{\mathsf T}\Htil_X$, then
$\vx\in\Row(\Htil_X)$ and
$H_Z\vx^{\mathsf T}=H_Z\Htil_X^{\mathsf T}\vlambda=0$,
so $\vx\in C_Z$; by definition of $\mathcal{L}_X$, we also have
$\vx\notin C_X^\perp$.
Hence every feasible coefficient vector produces an element of
$(C_Z\cap \Row(\Htil_X))\setminus C_X^\perp$.

Conversely, if $\vx\in (C_Z\cap \Row(\Htil_X))\setminus C_X^\perp$,
then $\vx\in\Row(\Htil_X)$, so there exists some
$\vlambda\in(\Ftwo^P)^{L/2-J}$ with
$\vx=\vlambda^{\mathsf T}\Htil_X$.
Because $\vx\in C_Z$,
$H_Z\Htil_X^{\mathsf T}\vlambda=H_Z\vx^{\mathsf T}=0$,
and since $\vx\notin C_X^\perp$, this $\vlambda$ belongs to $\mathcal{L}_X$.

Therefore the image set
\[
\{\vlambda^{\mathsf T}\Htil_X:\vlambda\in\mathcal{L}_X\}
\]
coincides exactly with
\[
(C_Z\cap \Row(\Htil_X))\setminus C_X^\perp.
\]
This proves the displayed formula for $d_X^{(\mathrm{lat})}$.
Finally, since
$d_X=\min\{d_X^{(\mathrm{lat})},d_X^{(\mathrm{nlat})}\}$,
every latent witness also yields the overall upper bound $d_X\le \wt(\vx)$.
The $Z$-side argument is identical.
\end{proof}

This proposition is the latent upper-bound statement used later in the paper.
Construction of the coefficient vectors is deferred to
Section~\ref{sec:results-en}, where the upper-bound criteria are applied to
explicit examples.

The exact latent lower-bound theory based on block-constant compression, together
with its rank-test and satisfiability (SAT) / satisfiability modulo theories (SMT) certification machinery, is placed in
Appendix~\ref{app:latent-lower-en}.  The main text next gives upper-bound
statements for restricted lifted subspaces.

\section{Restricted-Lift Structural Upper Bounds}
\label{sec:nonlatent-structural-en}

Section~\ref{sec:latent-en} treated logical representatives contained in the
latent row spaces.  From this section onward, we look for logical
representatives outside those latent row spaces; these give the non-latent side
of the upper-bound theory.  The common formulation is to restrict the search to
the image of a linear map into the full length-$LP$ ambient space.  A vector
found in such a subspace gives an upper bound only after it is checked to lie in
the appropriate parity-check kernel and outside the stabilizer row space.  The
general statement uses a linear map
\[
\Phi:\Ftwo^r\longrightarrow \Ftwo^{LP}.
\]
The following proposition records the search-independent condition:
a vector in this image that passes the kernel and row-space tests gives an
upper bound.  The three choices of $\Phi$ used in this paper are introduced below in the full-fiber,
fiber-quotient, and CRT-compression cases of
Subsections~\ref{sec:full-fiber-block-en}, \ref{sec:fiber-quotient-en}, and
\ref{sec:crt-compression-en}.

\begin{proposition}[Restricted-lift upper bound]
\label{prop:restricted-lift-upper-en}
Let $\Phi_X:\Ftwo^{r_X'}\to\Ftwo^{LP}$ and
$\Phi_Z:\Ftwo^{r_Z'}\to\Ftwo^{LP}$ be linear maps.  If
$\bar{\vx}\in\Ftwo^{r_X'}$ satisfies
\[
H_Z\Phi_X(\bar{\vx})^{\mathsf T}=0,
\qquad
\Phi_X(\bar{\vx})\notin\Row(H_X),
\]
then $\Phi_X(\bar{\vx})$ is a valid $X$ logical representative and
\[
d_X\le \wt(\Phi_X(\bar{\vx})).
\]
If in addition $\Phi_X(\bar{\vx})\notin\Row(\Htil_X)$, then
\[
d_X^{(\mathrm{nlat})}\le \wt(\Phi_X(\bar{\vx})).
\]
The analogous statement with $X$ and $Z$ interchanged also holds.
\end{proposition}

\begin{proof}
The $X$ side follows directly from the CSS distance definition:
$H_Z\Phi_X(\bar{\vx})^{\mathsf T}=0$ says that the lifted vector lies in
$C_Z$, while $\Phi_X(\bar{\vx})\notin\Row(H_X)=C_X^\perp$ says that it is not a
stabilizer.  Hence it represents a nontrivial $X$ logical operator.  The
additional condition $\Phi_X(\bar{\vx})\notin\Row(\Htil_X)$ places the same
representative outside the latent row space, giving the non-latent bound by
Definition~\ref{def:latent-nonlatent-distances-en}.  The $Z$ side is identical.
\end{proof}

We next treat the block-constant lift, the fiber-pattern lift, and the CRT
stripe subspace.  The first two are written explicitly as maps $\Phi$, while
the CRT case is written as a restricted subspace.  The symbol \drls is used only
as a collective label for these restricted-lift entries; the component entries
\dblk, \dfib, and \dcrt are reported separately.

\subsection{Block-Constant Quotient for Full-Fiber Block Compression}
\label{sec:classical-mblock-en}

Assume $P=mQ$ with $m\ge 2$. For $t\in\Z_Q$, define the coset
\[
[t]_m:=\{t,t+Q,t+2Q,\dots,t+(m-1)Q\}\subset\Z_P.
\]
A length-$P$ vector is called \emph{$m$-block constant} if it is constant on every coset $[t]_m$.

Let $U_m(P)\subset \Ftwo^P$ be the $m$-block-constant subspace. Define the compression and lift maps
\[
\pi_m:U_m(P)\to\Ftwo^Q,\qquad \iota_m:\Ftwo^Q\to U_m(P)
\]
by
\[
[\pi_m(\vx)]_t:=\vx_t,\qquad
[\iota_m(\bar{\vx})]_{t+jQ}:=\bar{\vx}_t.
\]
Applying these maps blockwise over $N$ blocks yields
\[
\pi_{m,N}:U_m(NP)\to\Ftwo^{NQ},\qquad
\iota_{m,N}:\Ftwo^{NQ}\to U_m(NP),
\]
with
\[
\wt(\iota_{m,N}(\bar{\vx}))=m\,\wt(\bar{\vx}).
\]
If the relevant matrices preserve block-constant structure, this construction
replaces a length-$P$ equation by a length-$Q$ quotient equation.  This is the
APM version of the quotient operation used for classical quasi-cyclic LDPC code constructions~\cite{Fossorier2004,MitchellSmarandacheCostello2014}.

\begin{lemma}[Descent of APMs on block-constant subspaces]
\label{lem:apm-descent-en}
Let $M$ be the $P\times P$ APM corresponding to $x\mapsto ax+b\pmod P$, and let $\bar M$ be the quotient APM obtained by reducing the same affine rule modulo $Q=P/m$. Then $M$ preserves $U_m(P)$. Moreover, for every $\bar{\vu}\in\Ftwo^Q$,
\[
M\iota_m(\bar{\vu})=\iota_m(\bar M\bar{\vu}),\qquad
\pi_m(M\iota_m(\bar{\vu}))=\bar M\bar{\vu}.
\]
The same commutativity holds blockwise for any block matrix made of APM blocks.
\end{lemma}

\begin{proof}
By definition of an APM,
\[
M\ve_t=\ve_{at+b}.
\]
For every $t\in\Z_Q$ and $s\in\{0,\dots,m-1\}$,
\[
a(t+sQ)+b\equiv at+b\pmod Q,
\]
so the affine permutation maps the coset $[t]_m$ to $[at+b]_m$. Hence the image of an $m$-block-constant vector is again $m$-block constant, and the action on cosets is exactly the quotient affine map modulo $Q$. The displayed identities follow directly from the definitions of $\pi_m$ and $\iota_m$.
\end{proof}

This lemma shows that block compression is an exact quotient/lift operation.
It is the algebraic statement underlying the CSS full-fiber
restricted-lift special case below: a compressed vector satisfying the quotient
check equation can be lifted to a block-constant vector satisfying the original
check equation.
The next subsection adds the CSS row-space tests needed to make the lifted vector a logical representative.

\subsection{Full-Fiber Block-Compression Special Case}
\label{sec:nonlatent-en}
\label{sec:full-fiber-block-en}
This subsection obtains non-latent upper bounds from the block-constant lift.
The additional non-latent test is literal: on the $X$ side the lifted vector
must be outside $\Row(\Htil_X)$, and on the $Z$ side it must be outside
$\Row(\Htil_Z)$.

Fix $m\mid P$ and consider the same subspace
\[
U_m(LP)\subset \Ftwo^{LP}.
\]
Since each APM descends modulo $Q=P/m$, the active matrices descend to compressed check matrices
\[
\bar H_X^{(m)},\bar H_Z^{(m)}\in\Ftwo^{JQ\times LQ}.
\]
Before using them, one must verify that the active check equations are preserved by compression.

\begin{lemma}[Compression equivalence for active checks]
\label{lem:nonlatent-compression-en}
For every $\vx\in U_m(LP)$ and $\bar{\vx}=\pi_{m,L}(\vx)$,
\[
H_Z\vx^{\mathsf T}=0
\iff
\bar H_Z^{(m)}\bar{\vx}^{\mathsf T}=0.
\]
Similarly, for every $\vz\in U_m(LP)$ and $\bar{\vz}=\pi_{m,L}(\vz)$,
\[
H_X\vz^{\mathsf T}=0
\iff
\bar H_X^{(m)}\bar{\vz}^{\mathsf T}=0.
\]
Moreover,
\[
\wt(\vx)=m\,\wt(\bar{\vx}),\qquad \wt(\vz)=m\,\wt(\bar{\vz}).
\]
\end{lemma}

\begin{proof}
Apply Lemma~\ref{lem:apm-descent-en} blockwise to the quotient blocks of the active matrices. The active check vector of a lifted block-constant vector is exactly the lifted check vector computed by the compressed matrix. The weight formula follows from the lift.
\end{proof}

This lemma states that a quotient vector satisfying the compressed kernel
condition lifts to a vector satisfying the original active-check kernel
condition exactly.

\begin{proposition}[Full-fiber block-compression special case]
\label{prop:nonlatent-upper-en}
Let $m\mid P$ and $Q=P/m$. Take $\bar{\vx}\in\Ker(\bar H_Z^{(m)})$ and let $\vx=\iota_{m,L}(\bar{\vx})$. If $\vx\notin \Row(H_X)$ and $\vx\notin \Row(\Htil_X)$, then $\vx$ is a valid $X$ logical representative outside the latent row space, and $d_X\le \wt(\vx)=m\,\wt(\bar{\vx})$. Similarly, if $\bar{\vz}\in\Ker(\bar H_X^{(m)})$ lifts to $\vz=\iota_{m,L}(\bar{\vz})$ with $\vz\notin \Row(H_Z)\cup \Row(\Htil_Z)$, then $d_Z\le \wt(\vz)=m\,\wt(\bar{\vz})$.
\end{proposition}

\begin{proof}
By Lemma~\ref{lem:nonlatent-compression-en}, the lift $\vx$ satisfies $H_Z\vx^{\mathsf T}=0$, so $\vx\in C_Z$. If additionally $\vx\notin \Row(H_X)=C_X^\perp$, then $\vx$ is a valid $X$ logical representative. The extra condition $\vx\notin \Row(\Htil_X)$ ensures that it lies outside the latent row space. Hence $d_X\le \wt(\vx)=m\,\wt(\bar{\vx})$. The $Z$ case is the same.
\end{proof}

Concrete use of Proposition~\ref{prop:nonlatent-upper-en} is explained in
Section~\ref{sec:results-en}.

The block-constant lift is only the first member of the restricted-lift
family.  The next two subsections keep
Proposition~\ref{prop:restricted-lift-upper-en} as the underlying \markededit{sufficient condition}
while enlarging the allowed lifted subspace.

\subsection{Fiber-Quotient Special Case}
\label{sec:fiber-quotient-en}
The fiber-quotient bound lifts each quotient bit only to a selected fiber
pattern.  \markededit{The proposition below is stated for every nonempty pattern; in the
tables, the notation \dfib is reserved for proper patterns
\(\emptyset\neq S\subsetneq\Z/m\Z\), because the full pattern is already
reported as \dblk.}

\begin{definition}[Fiber-pattern lift and fiber-quotient subspace]
\label{def:fiber-quotient-subspace-en}
Let $m\mid P$, put $Q=P/m$, and fix a nonempty pattern
$S\subseteq \Z/m\Z$.  Define the single-block fiber-pattern lift
\[
\phi_{m,S}:\Ftwo^Q\longrightarrow \Ftwo^P
\]
by
\[
(\phi_{m,S}(\vy))_{t+kQ}
:=
\begin{cases}
y_t, & k\in S,\\
0, & k\notin S,
\end{cases}
\qquad
t\in\Z_Q,\ k\in\Z/m\Z .
\]
The corresponding length-$LP$ lift is the blockwise direct sum
\[
\Phi_{m,S,L}:=\bigoplus_{j=0}^{L-1}\phi_{m,S}:
\Ftwo^{LQ}\longrightarrow\Ftwo^{LP}.
\]
Its image
\[
\mathcal F_{m,S}:=\operatorname{im}(\Phi_{m,S,L})
\subseteq \Ftwo^{LP}
\]
is called the fiber-quotient subspace of pattern $S$.
\end{definition}

In particular, when $S=\Z/m\Z$, one has $\mathcal F_{m,S}=U_m(LP)$.
Thus block compression is the full-fiber special case of fiber quotient.

Proper patterns $S\subsetneq\Z/m\Z$ give additional candidates that are
supported only on selected fibers, with up to $2^m-1$ candidate patterns for
each divisor $m$.  If $M_{a,b}$ denotes the APM for
$x\mapsto ax+b\pmod P$, then in general
\[
M_{a,b}\operatorname{im}(\phi_{m,S})\not\subseteq \operatorname{im}(\phi_{m,S}).
\]
The method therefore restricts the original active checks to the image of $\Phi_{m,S,L}$:
\[
H_{Z}^{(m,S)}:=H_Z\Phi_{m,S,L},\qquad
H_{X}^{(m,S)}:=H_X\Phi_{m,S,L}.
\]
These are ordinary binary matrices with $LQ$ columns.  They exactly encode the
parity-check products of lifted fiber-pattern vectors.

\begin{lemma}[Check restriction for fiber quotients]
\label{lem:fiber-quotient-syndrome-en}
For every $\bar{\vx}\in\Ftwo^{LQ}$ and
$\vx=\Phi_{m,S,L}(\bar{\vx})$,
\[
H_Z\vx^{\mathsf T}=0
\iff
H_Z^{(m,S)}\bar{\vx}^{\mathsf T}=0.
\]
Similarly, for every $\bar{\vz}\in\Ftwo^{LQ}$ and
$\vz=\Phi_{m,S,L}(\bar{\vz})$,
\[
H_X\vz^{\mathsf T}=0
\iff
H_X^{(m,S)}\bar{\vz}^{\mathsf T}=0.
\]
Moreover,
\[
\wt(\Phi_{m,S,L}(\vy))=|S|\,\wt(\vy)
\qquad(\vy\in\Ftwo^{LQ}).
\]
\end{lemma}

\begin{proof}
The first two equivalences are just the definitions of
$H_Z^{(m,S)}$ and $H_X^{(m,S)}$ as restrictions of the original active-check
matrices to the image of $\Phi_{m,S,L}$.
For the weight identity, the sets
\[
\{t+kQ:\ k\in S\}\subseteq \Z_P
\]
are disjoint for distinct quotient coordinates $t$, and different block
columns are disjoint as well.  Hence every nonzero quotient coordinate
contributes exactly $|S|$ ones after lifting.
\end{proof}

\begin{proposition}[Fiber-quotient special case]
\label{prop:fiber-quotient-upper-en}
Let $m\mid P$, $Q=P/m$, and $\emptyset\ne S\subseteq\Z/m\Z$.
If
\[
\bar{\vx}\in\Ker(H_Z^{(m,S)})
\quad\text{and}\quad
\vx:=\Phi_{m,S,L}(\bar{\vx})\notin\Row(H_X),
\]
then $\vx$ is a valid $X$ logical representative and
\[
d_X\le \wt(\vx)=|S|\,\wt(\bar{\vx}).
\]
If in addition $\vx\notin\Row(\Htil_X)$, then
$d_X^{(\mathrm{nlat})}\le |S|\,\wt(\bar{\vx})$.

Similarly, if
\[
\bar{\vz}\in\Ker(H_X^{(m,S)})
\quad\text{and}\quad
\vz:=\Phi_{m,S,L}(\bar{\vz})\notin\Row(H_Z),
\]
then
\[
d_Z\le \wt(\vz)=|S|\,\wt(\bar{\vz}),
\]
and the additional condition $\vz\notin\Row(\Htil_Z)$ gives
$d_Z^{(\mathrm{nlat})}\le |S|\,\wt(\bar{\vz})$.
\end{proposition}

\begin{proof}
We prove the $X$ side.  By Lemma~\ref{lem:fiber-quotient-syndrome-en},
$\bar{\vx}\in\Ker(H_Z^{(m,S)})$ implies
$H_Z\vx^{\mathsf T}=0$, so $\vx\in C_Z$.
If $\vx\notin\Row(H_X)=C_X^\perp$, then $\vx$ is a nontrivial $X$ logical
representative.  Therefore the CSS distance definition gives
$d_X\le\wt(\vx)$, and the weight identity in the lemma gives
$\wt(\vx)=|S|\,\wt(\bar{\vx})$.
The non-latent statement follows from
Definition~\ref{def:latent-nonlatent-distances-en} after adding
$\vx\notin\Row(\Htil_X)$.  The $Z$ side is identical.
\end{proof}

This completes the fiber-quotient special case of the restricted-lift
upper-bound method.  The practical selection of patterns and quotient
vectors is discussed only in Section~\ref{sec:results-en}.

\subsection{CRT-Compression Special Case}
\label{sec:crt-compression-en}
CRT-compression restricts the search space to a stripe subspace coming from CRT
coordinates.  If $P=q_1q_2$ with $\gcd(q_1,q_2)=1$, then
\[
\Z_P \cong \Z_{q_1}\times \Z_{q_2},
\qquad
t\longmapsto (t\bmod q_1,\ t\bmod q_2).
\]

\begin{definition}[CRT stripe subspace]
\label{def:crt-stripe-subspace-en}
Assume that $P=q_1q_2$ and $\gcd(q_1,q_2)=1$.
Define the single-block CRT stripe subspace by
\[
\mathcal{S}^{(1)}_{q_1,q_2}
:=
\operatorname{span}\!\Bigl(
\{\mathbf{1}_{\{x\in\Z_P:\,x\equiv r\!\!\!\pmod{q_1}\}}:0\le r<q_1\}
\cup
\{\mathbf{1}_{\{x\in\Z_P:\,x\equiv s\!\!\!\pmod{q_2}\}}:0\le s<q_2\}
\Bigr)
\subseteq \Ftwo^{P}.
\]
The full-length CRT stripe subspace is then the blockwise direct sum
\[
\mathcal{S}^{\mathrm{crt}}_{q_1,q_2}
:=
\bigoplus_{j=0}^{L-1}\mathcal{S}^{(1)}_{q_1,q_2}
\subseteq \Ftwo^{LP}.
\]
\end{definition}

This definition gives a linear subspace larger than the block-constant subspace
and smaller than the full ambient space.  CRT-compression checks the kernel and
row-space conditions inside this subspace.

\begin{proposition}[CRT-compression special case]
\label{prop:crt-compression-upper-en}
Assume that $P=q_1q_2$ and $\gcd(q_1,q_2)=1$.
Define
\[
u_X^{\mathrm{crt}}(q_1,q_2)
:=
\min\left\{
\wt(\vx):
\vx \in \mathcal{S}^{\mathrm{crt}}_{q_1,q_2}\cap \Ker(H_Z)\setminus \Row(H_X)
\right\},
\]
\[
u_Z^{\mathrm{crt}}(q_1,q_2)
:=
\min\left\{
\wt(\vz):
\vz \in \mathcal{S}^{\mathrm{crt}}_{q_1,q_2}\cap \Ker(H_X)\setminus \Row(H_Z)
\right\},
\]
with the convention that the minimum is $+\infty$ when the corresponding set is empty.
Then $d_X\le u_X^{\mathrm{crt}}(q_1,q_2)$ and
$d_Z\le u_Z^{\mathrm{crt}}(q_1,q_2)$.

In particular, if
$\vx \in \mathcal{S}^{\mathrm{crt}}_{q_1,q_2}\cap \Ker(H_Z)\setminus \Row(H_X)$,
then $d_X\le \wt(\vx)$. If in addition
$\vx\notin \Row(\Htil_X)$, then also
$d_X^{(\mathrm{nlat})}\le \wt(\vx)$.
Likewise, if
$\vz \in \mathcal{S}^{\mathrm{crt}}_{q_1,q_2}\cap \Ker(H_X)\setminus \Row(H_Z)$,
then $d_Z\le \wt(\vz)$. If moreover
$\vz\notin \Row(\Htil_Z)$, then
$d_Z^{(\mathrm{nlat})}\le \wt(\vz)$.
\end{proposition}

\begin{proof}
We treat the $X$ side. Set
\[
\mathcal{W}_X(q_1,q_2)
:=
\mathcal{S}^{\mathrm{crt}}_{q_1,q_2}\cap \Ker(H_Z)\setminus \Row(H_X).
\]
Every vector $\vx\in\mathcal{W}_X(q_1,q_2)$ is, by construction, a valid
$X$-type non-stabilizer logical representative. Therefore the CSS definition
implies $d_X\le \wt(\vx)$ for every $\vx\in\mathcal{W}_X(q_1,q_2)$.
If the set is nonempty, taking the minimum over it yields
\[
d_X\le \min_{\vx\in\mathcal{W}_X(q_1,q_2)}\wt(\vx)
=u_X^{\mathrm{crt}}(q_1,q_2).
\]
If the set is empty, the convention $u_X^{\mathrm{crt}}(q_1,q_2)=+\infty$
makes the inequality trivial.

If in addition $\vx\notin\Row(\Htil_X)$, then
$\vx\in C_Z\setminus C_X^\perp$ and $\vx\notin \Row(\Htil_X)$,
so Definition~\ref{def:latent-nonlatent-distances-en} gives
$d_X^{(\mathrm{nlat})}\le \wt(\vx)$.
The $Z$ side is identical.
\end{proof}

This completes the CRT-compression special case of the restricted-lift
upper-bound method.

\section{Additional Witness-Based Upper Bounds}
\label{sec:local-decoding-en}

This section treats bounds that do not use a restricted lift.  The search
procedure may generate candidate vectors or supports in any way.  The
upper-bound statement starts only after the resulting vector is checked to lie
in the appropriate parity-check kernel and outside the stabilizer row space.

\subsection{The Direct CSS-Search Upper Bound}
\label{sec:direct-css-search-en}
The direct CSS-search bound records the smallest weight among a finite set of
already verified CSS logical vectors.

\begin{proposition}[Direct CSS-search upper bound]
\label{prop:direct-css-upper-en}
Let
\[
\mathcal{W}_X\subseteq \Ker(H_Z)\setminus \Row(H_X),\qquad
\mathcal{W}_Z\subseteq \Ker(H_X)\setminus \Row(H_Z)
\]
be any finite candidate sets.
If they are nonempty, define
\[
u_X^{\mathrm{dir}}:=\min_{\vx\in\mathcal{W}_X}\wt(\vx),\qquad
u_Z^{\mathrm{dir}}:=\min_{\vz\in\mathcal{W}_Z}\wt(\vz).
\]
Then $d_X\le u_X^{\mathrm{dir}}$ and $d_Z\le u_Z^{\mathrm{dir}}$.

In particular, if $\vx\in\Ker(H_Z)\setminus \Row(H_X)$, then
$d_X\le \wt(\vx)$. Likewise, if
$\vz\in\Ker(H_X)\setminus \Row(H_Z)$, then
$d_Z\le \wt(\vz)$.
Moreover, if some $\vx\in\mathcal{W}_X$ also satisfies
$\vx\notin \Row(\Htil_X)$, then
$d_X^{(\mathrm{nlat})}\le \wt(\vx)$, and analogously
$\vz\in\mathcal{W}_Z$ with $\vz\notin \Row(\Htil_Z)$ implies
$d_Z^{(\mathrm{nlat})}\le \wt(\vz)$.
\end{proposition}

\begin{proof}
We treat the $X$ side. Every vector in $\mathcal{W}_X$ belongs to
$\Ker(H_Z)\setminus \Row(H_X)$, so every such vector is already an $X$-type
non-stabilizer logical representative. Therefore
$d_X\le \wt(\vx)$ for every $\vx\in\mathcal{W}_X$.
If $\mathcal{W}_X$ is nonempty, taking the minimum over the sampled candidate
set gives
\[
d_X\le \min_{\vx\in\mathcal{W}_X}\wt(\vx)=u_X^{\mathrm{dir}}.
\]
The one-vector statement is the special case $\mathcal{W}_X=\{\vx\}$.

If moreover $\vx\notin\Row(\Htil_X)$, then
$\vx\in C_Z\setminus C_X^\perp$ and $\vx\notin \Row(\Htil_X)$,
so Definition~\ref{def:latent-nonlatent-distances-en} yields
$d_X^{(\mathrm{nlat})}\le \wt(\vx)$.
The $Z$ side is identical.
\end{proof}

The proposition \markededit{follows directly} from the CSS
definition.  Any procedure may produce the candidate sets
$\mathcal W_X,\mathcal W_Z$, but the upper-bound claim begins only after the
CSS membership and non-stabilizer conditions have been \markededit{checked}.

\subsection{The Cycle-8 ETS Upper Bound}
\label{sec:cycle8-ets-en}
All codes considered here have girth $8$, so cycle-8 structures provide the
local supports \markededit{used below}.  We use trapping set terminology
following~\cite{KarimiBanihashemi2012,HashemiBanihashemi2016}, only to \markededit{identify}
support vectors in the parity-check kernel.

\begin{definition}[Cycle-8-connected ETS]
\label{def:cycle8-ets-en}
Let $H\in\Ftwo^{m\times n}$ and let $G(H)$ be the corresponding Tanner graph. For a variable-node set $S\subseteq V(H)$, define
\[
\Gamma(S):=\{c\in C(H):N(c)\cap S\neq\emptyset\},
\qquad
d_S(c):=|N(c)\cap S|.
\]
Define the odd-check boundary by
\[
\partial S:=\{c\in\Gamma(S):d_S(c)\equiv 1\pmod 2\}.
\]
If every $c\in\Gamma(S)$ satisfies
\[
d_S(c)\in\{1,2\}
\]
and the induced subgraph on $S\cup\Gamma(S)$ is connected, then $S$ is an ETS of type $(a,b)$, where
\[
a:=|S|,\qquad b:=|\partial S|.
\]
If, in addition, $S$ can be written as a union of simple $8$-cycles such that each new cycle shares at least one variable node or check node with the previous union, then $S$ is called a \emph{cycle-8-connected ETS}.
\end{definition}

Figure~\ref{fig:cycle8-connected-ets-en} recalls representative
cycle-8-connected ETS patterns from~\cite{Kasai2026}.  These patterns are
examples rather than part of the definition: the definition only requires that
the support be grown by adding one 8-cycle at a time while maintaining overlap
with the union already constructed.

\begin{figure}[t]
\centering
\begin{minipage}{0.48\linewidth}
\centering
\begin{tikzpicture}[
  x=0.6cm,y=0.6cm,
  var/.style={circle,draw=black,fill=white,line width=0.6pt,minimum size=5.5pt,inner sep=0pt},
  chk/.style={rectangle,draw=black,fill=gray!20,line width=0.6pt,minimum size=5.2pt,inner sep=0pt,rounded corners=1.2pt},
  odd/.style={rectangle,draw=black,fill=gray!50,line width=0.6pt,minimum size=5.2pt,inner sep=0pt,rounded corners=1.2pt},
  edge/.style={line width=0.6pt,draw=black,line cap=round,line join=round}
]
\node[var] (v1) at (0,0) {};
\node[var] (v2) at (0,2) {};
\node[var] (v3) at (2,2) {};
\node[var] (v4) at (2,0) {};
\node[var] (v5) at (4,2) {};
\node[var] (v6) at (4,0) {};

\node[chk] (c1) at (1,2) {};
\node[chk] (c2) at (2,1) {};
\node[chk] (c3) at (1,0) {};
\node[chk] (c4) at (0,1) {};
\node[chk] (c5) at (3,2) {};
\node[chk] (c6) at (4,1) {};
\node[chk] (c7) at (3,0) {};
\node[chk] (c8) at (1.0,1.0) {};

\node[odd] (o1) at (-0.8,0) {};
\node[odd] (o2) at (4.8,2) {};

\draw[edge] (v2)--(c1)--(v3)--(c2)--(v4)--(c3)--(v1)--(c4)--(v2);
\draw[edge] (v3)--(c5)--(v5)--(c6)--(v6)--(c7)--(v4)--(c2);
\draw[edge] (v1)--(o1);
\draw[edge] (v5)--(o2);
\draw[edge] (v2) to[bend right=40] (c8.west);
\draw[edge] (c8.east) to[bend left=50] (v6);
\end{tikzpicture}
\par
$(6,2)$ ETS (two connected 8-cycles)
\end{minipage}
\hfill
\begin{minipage}{0.48\linewidth}
\centering
\begin{tikzpicture}[
  x=0.5cm,y=0.6cm,
  var/.style={circle,draw=black,fill=white,line width=0.6pt,minimum size=5.2pt,inner sep=0pt},
  chk/.style={rectangle,draw=black,fill=gray!20,line width=0.6pt,minimum size=5.0pt,inner sep=0pt,rounded corners=1.2pt},
  odd/.style={rectangle,draw=black,fill=gray!50,line width=0.6pt,minimum size=5.0pt,inner sep=0pt,rounded corners=1.2pt},
  edge/.style={line width=0.6pt,draw=black,line cap=round,line join=round}
]
\node[var] (v1) at (0,2) {};
\node[var] (v2) at (0,0) {};
\node[var] (v3) at (2,2) {};
\node[var] (v4) at (2,0) {};
\node[var] (v5) at (4,2) {};
\node[var] (v6) at (4,0) {};
\node[var] (v7) at (6,2) {};
\node[var] (v8) at (6,0) {};
\node[var] (v9) at (8,2) {};
\node[var] (v10) at (8,0) {};
\node[var] (v11) at (10,2) {};
\node[var] (v12) at (10,0) {};

\node[chk] (u0) at (0,1) {};
\node[chk] (u1) at (2,1) {};
\node[chk] (u2) at (4,1) {};
\node[chk] (u3) at (6,1) {};
\node[chk] (u4) at (8,1) {};
\node[chk] (u5) at (10,1) {};

\node[chk] (t1) at (1,2) {};
\node[chk] (t2) at (3,2) {};
\node[chk] (t3) at (5,2) {};
\node[chk] (t4) at (7,2) {};
\node[chk] (t5) at (9,2) {};
\node[chk] (b1) at (1,0) {};
\node[chk] (b2) at (3,0) {};
\node[chk] (b3) at (5,0) {};
\node[chk] (b4) at (7,0) {};
\node[chk] (b5) at (9,0) {};
\node[chk] (u6) at (5,1.0) {};

\node[odd] (o3) at (-1,0) {};
\node[odd] (o4) at (11,2) {};

\draw[edge] (v1)--(t1)--(v3)--(u1)--(v4)--(b1)--(v2)--(u0)--(v1);
\draw[edge] (v3)--(t2)--(v5)--(u2)--(v6)--(b2)--(v4)--(u1)--(v3);
\draw[edge] (v5)--(t3)--(v7)--(u3)--(v8)--(b3)--(v6)--(u2)--(v5);
\draw[edge] (v7)--(t4)--(v9)--(u4)--(v10)--(b4)--(v8)--(u3)--(v7);
\draw[edge] (v9)--(t5)--(v11)--(u5)--(v12)--(b5)--(v10)--(u4)--(v9);
\draw[edge] (v2)--(o3);
\draw[edge] (v11)--(o4);
\draw[edge] (v1) to[bend right=40] (u6.west);
\draw[edge] (u6.east) to[bend left=40] (v12);
\end{tikzpicture}
\par
$(12,2)$ ETS (five connected 8-cycles)
\end{minipage}
\par\medskip
\begin{minipage}{0.70\linewidth}
\centering
\begin{tikzpicture}[
  x=0.6cm,y=0.6cm,
  var/.style={circle,draw=black,fill=white,line width=0.6pt,minimum size=5.4pt,inner sep=0pt},
  chk/.style={rectangle,draw=black,fill=gray!20,line width=0.6pt,minimum size=5.1pt,inner sep=0pt,rounded corners=1.2pt},
  odd/.style={rectangle,draw=black,fill=gray!50,line width=0.6pt,minimum size=5.1pt,inner sep=0pt,rounded corners=1.2pt},
  edge/.style={line width=0.6pt,draw=black,line cap=round,line join=round}
]
\node[var] (v1) at (0,0) {};
\node[var] (v2) at (0,2) {};
\node[var] (v3) at (2,2) {};
\node[var] (v4) at (2,0) {};
\node[var] (v5) at (4,2) {};
\node[var] (v6) at (4,0) {};

\node[chk] (c1) at (1,2) {};
\node[chk] (c2) at (2,1) {};
\node[chk] (c3) at (1,0) {};
\node[chk] (c4) at (0,1) {};
\node[chk] (c5) at (3,2) {};
\node[chk] (c6) at (4,1) {};
\node[chk] (c7) at (3,0) {};
\node[chk] (c8) at (1.0,1.0) {};

\node[chk] (c599) at (-0.8,0.8) {};
\node[chk] (c542) at (4.8,0.8) {};

\node[var] (v8) at (-1.8,2.2) {};
\node[chk] (c893) at (1.6,3.2) {};
\node[var] (v7) at (5.0,2.2) {};

\node[odd] (o1715) at (-2.8,3.6) {};
\node[odd] (o1991) at (6.8,3.6) {};

\draw[edge] (v2)--(c1)--(v3)--(c2)--(v4)--(c3)--(v1)--(c4)--(v2);
\draw[edge] (v3)--(c5)--(v5)--(c6)--(v6)--(c7)--(v4)--(c2);
\draw[edge] (v1)--(c599);
\draw[edge] (v5)--(c542);
\draw[edge] (v2) to[bend right=40] (c8.west);
\draw[edge] (c8.east) to[bend left=40] (v6);

\draw[edge] (c542)--(v7)--(c893)--(v8)--(c599);
\draw[edge] (v8)--(o1715);
\draw[edge] (v7)--(o1991);
\end{tikzpicture}
\par
$(8,2)$ ETS (path-4 attachment)
\end{minipage}
\caption{Representative cycle-8-connected ETS patterns used in the library
of~\cite{Kasai2026}.  White circles are variable nodes, light squares are even
checks, and dark squares are odd checks.}
\label{fig:cycle8-connected-ets-en}
\end{figure}
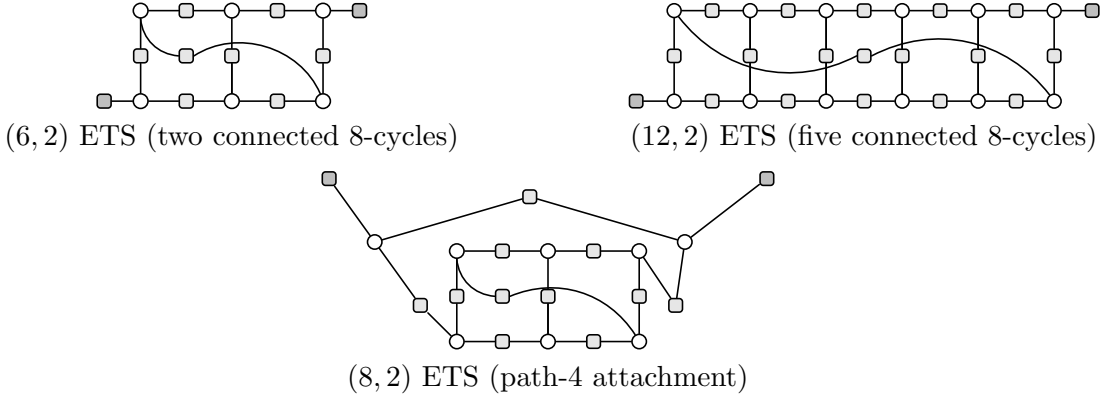

For girth $8$, cycle-based local support candidates start from $8$-cycles.  The
paper uses this class only to state the
parity-check-kernel and distance-upper-bound implications below.

\begin{lemma}[The odd-check boundary equals the induced check vector]
\label{lem:cycle8-boundary-syndrome-en}
For any $S\subseteq V(H)$, let $\mathbf{1}_S\in\Ftwo^n$ be its indicator vector. Then
\[
H\,\mathbf{1}_S^{\mathsf T}=\mathbf{1}_{\partial S}^{\mathsf T}.
\]
In particular, if $S$ is an ETS of type $(a,b)$, then the induced check vector has weight exactly $b$.
\end{lemma}

\begin{proof}
Each component of $H\mathbf{1}_S^{\mathsf T}$ is the parity of the number of neighbors of the corresponding check node inside $S$, namely the parity of $d_S(c)$. This is $1$ exactly when $c\in\partial S$.
\end{proof}

This lemma identifies the graph-theoretic boundary with a parity-check product.
The two propositions below use only the cases needed for the upper
bound: $b=0$ gives a support vector in $\Ker(H)$, and two $b=2$ ETSs with the
same boundary give a kernel vector by a single symmetric difference.

\begin{proposition}[Classical upper bounds from cycle-8-connected ETSs with $b=0$]
\label{prop:cycle8-b0-upper-en}
If $S$ is a cycle-8-connected ETS of type $(a,0)$, then
$\mathbf{1}_S\in\Ker(H)$,
hence the corresponding classical code $C(H)=\Ker(H)$ satisfies $d(C(H))\le a$.
\end{proposition}

\begin{proof}
By Lemma~\ref{lem:cycle8-boundary-syndrome-en},
$\partial S=\emptyset$ implies $H\,\mathbf{1}_S^{\mathsf T}=0$.
\end{proof}

\begin{proposition}[Classical upper bounds from pairs of $(a,2)$ ETSs with the same boundary]
\label{prop:cycle8-b2pair-upper-en}
Let $S_1,S_2$ be cycle-8-connected ETSs of types $(a_1,2)$ and $(a_2,2)$, respectively. If
$\partial S_1=\partial S_2$ and $S_1\triangle S_2\neq\emptyset$,
then
$\mathbf{1}_{S_1}+\mathbf{1}_{S_2}=\mathbf{1}_{S_1\triangle S_2}\in \Ker(H)$,
and therefore $d(C(H))\le |S_1\triangle S_2|\le a_1+a_2$.
\end{proposition}

\begin{proof}
Lemma~\ref{lem:cycle8-boundary-syndrome-en} gives
$H\,\mathbf{1}_{S_i}^{\mathsf T}=\mathbf{1}_{\partial S_i}^{\mathsf T}$ for $i=1,2$.
Since $\partial S_1=\partial S_2$, adding the two equations over $\Ftwo$ yields
$H(\mathbf{1}_{S_1}+\mathbf{1}_{S_2})^{\mathsf T}=0$.
The sum of indicator vectors equals the indicator vector of the symmetric difference. The
assumption $S_1\triangle S_2\neq\emptyset$ guarantees that this codeword is nonzero.
\end{proof}

To use the following corollary, one converts a cycle-8 ETS \markededit{support} into a
support vector, checks membership in the relevant parity-check kernel and
non-membership in the opposite row space, and uses the support size as the
upper bound.

\begin{corollary}[CSS upper bounds from cycle-8-connected ETS witnesses]
\label{cor:cycle8-css-upper-en}
Let $H=H_Z$, and let $\vx\in\Ker(H_Z)$ be a nonzero support vector obtained from Proposition~\ref{prop:cycle8-b0-upper-en} or Proposition~\ref{prop:cycle8-b2pair-upper-en}. If $\vx\notin \Row(H_X)$, then $\vx$ is a valid $X$ logical representative and $d_X\le \wt(\vx)$.
Similarly, if $\vz\in\Ker(H_X)\setminus \Row(H_Z)$ is a nonzero support vector obtained on the $Z$ side, then $d_Z\le \wt(\vz)$.
\end{corollary}

\begin{proof}
This is immediate from the CSS distance definitions together with $\Row(H_X)=C_X^\perp$ and $\Row(H_Z)=C_Z^\perp$.
\end{proof}

\subsection{The Decoder-Failure Upper Bound}
The decoder-failure bound \markededit{uses} a residual produced by a failed
decoding event.  For a true error $\mathbf e$ and a syndrome-matched estimate
$\hat{\mathbf e}$, set
\[
\mathbf{\Delta}=\mathbf{e}+\hat{\mathbf{e}}.
\]
If this residual is pure $X$ or pure $Z$ and is outside the corresponding
stabilizer row space, its weight is an upper bound.

\begin{proposition}
\label{prop:decoder-failure-upper-en}
Let $\mathbf{e}$ be the true error and $\hat{\mathbf{e}}$ an estimated error
with the same observed syndrome, and set the residual to
$\mathbf{\Delta}=\mathbf{e}+\hat{\mathbf{e}}$.
If $\mathbf{\Delta}$ is pure $X$-type and $\mathbf{\Delta}\notin\Row(H_X)$,
then $d_X\le \wt(\mathbf{\Delta})$.
Likewise, if $\mathbf{\Delta}$ is pure $Z$-type and
$\mathbf{\Delta}\notin\Row(H_Z)$, then $d_Z\le \wt(\mathbf{\Delta})$.
\end{proposition}
\begin{proof}
Since $\mathbf{e}$ and $\hat{\mathbf{e}}$ produce the same syndrome, the
residual $\mathbf{\Delta}$ has zero check syndrome.
Hence, if it is pure $X$-type then $\mathbf{\Delta}\in\Ker(H_Z)$, and if it is
pure $Z$-type then $\mathbf{\Delta}\in\Ker(H_X)$.
By assumption it lies outside the corresponding stabilizer row space, so it is
a nontrivial logical representative and its weight upper-bounds the
corresponding logical distance.
\end{proof}

In the present paper, Proposition~\ref{prop:decoder-failure-upper-en} is used
only to state the implication from a \markededit{checked decoder residual} to a distance
upper bound.

\section{Use of the Upper-Bound Methods}
\label{sec:results-en}
This section applies the upper-bound criteria to representative codes.
Detailed numerical tables and parameter rows are maintained on the supplementary
website~\cite{KasaiWeb2026}.  The website is updated separately from this
manuscript and records additional candidates and verification status; the main
text fixes several worked examples and records the $P\le 768$ rows used here
from the cited online table.
Table~\ref{tab:p768-apm-en} gives their affine maps in this order:
$C_1$ and $C_2$ are $P=216$ codes used for ETS and decoder-failure examples,
$C_3,\ldots,C_8$ are the online-table codes with $240\le P\le 576$,
$C_9$ is the printed $P=768$ code of~\cite{Kasai2026}, and
$C_{10}$ is a $P=768$ reference code.
Table~\ref{tab:p768-representative-en} records the corresponding upper-bound values.
The rows are ordered by increasing code length, equivalently by increasing $P$.
For \dets, a finite table entry \(w^{[s]}\) means that the reported ETS check
found a witness of weight \(w\) using candidate unions generated from at most
\(s\) connected \(8\)-cycles.

\begin{table}[t]
\centering
\caption{Affine data for the representative codes. Other code lists are available in~\cite{KasaiWeb2026}.}
\label{tab:p768-apm-en}
\scriptsize
\setlength{\tabcolsep}{2pt}
\renewcommand{\arraystretch}{1.08}
\begin{tabular}{ccc*{6}{>{\centering\arraybackslash}p{5.4em}}}
\toprule
Code & $P$ & map & $0$ & $1$ & $2$ & $3$ & $4$ & $5$ \\
\midrule
\multirow{2}{*}{$C_1$}
& \multirow{2}{*}{216}
& $f_i$ & \affexpr{73}{87} & \affexpr{199}{75} & \affexpr{181}{24} & \affexpr{127}{213} & \affexpr{199}{75} & \affexpr{181}{186} \\
& & $g_i$ & \affexpr{97}{188} & \affexpr{73}{204} & \affexpr{25}{140} & \affexpr{13}{154} & \affexpr{73}{168} & \affexpr{25}{68} \\
\midrule
\multirow{2}{*}{$C_2$}
& \multirow{2}{*}{216}
& $f_i$ & \affexpr{181}{105} & \affexpr{91}{57} & \affexpr{37}{192} & \affexpr{181}{186} & \affexpr{91}{57} & \affexpr{37}{156} \\
& & $g_i$ & \affexpr{193}{112} & \affexpr{97}{116} & \affexpr{193}{64} & \affexpr{181}{54} & \affexpr{97}{44} & \affexpr{193}{100} \\
\midrule
\multirow{2}{*}{$C_3$}
& \multirow{2}{*}{240}
& $f_i$ & \affexpr{1}{54} & \affexpr{31}{141} & \affexpr{1}{0} & \affexpr{31}{153} & \affexpr{151}{201} & \affexpr{1}{96} \\
& & $g_i$ & \affexpr{161}{112} & \affexpr{121}{148} & \affexpr{121}{180} & \affexpr{221}{162} & \affexpr{121}{164} & \affexpr{41}{44} \\
\midrule
\multirow{2}{*}{$C_4$}
& \multirow{2}{*}{264}
& $f_i$ & \affexpr{67}{126} & \affexpr{1}{177} & \affexpr{1}{126} & \affexpr{133}{198} & \affexpr{1}{261} & \affexpr{1}{144} \\
& & $g_i$ & \affexpr{89}{136} & \affexpr{1}{248} & \affexpr{89}{48} & \affexpr{89}{86} & \affexpr{1}{32} & \affexpr{1}{124} \\
\midrule
\multirow{2}{*}{$C_5$}
& \multirow{2}{*}{288}
& $f_i$ & \affexpr{73}{98} & \affexpr{241}{244} & \affexpr{241}{36} & \affexpr{1}{192} & \affexpr{1}{64} & \affexpr{97}{72} \\
& & $g_i$ & \affexpr{109}{75} & \affexpr{217}{186} & \affexpr{1}{120} & \affexpr{217}{0} & \affexpr{217}{126} & \affexpr{37}{141} \\
\midrule
\multirow{2}{*}{$C_6$}
& \multirow{2}{*}{384}
& $f_i$ & \affexpr{181}{282} & \affexpr{145}{8} & \affexpr{233}{180} & \affexpr{169}{84} & \affexpr{337}{296} & \affexpr{89}{108} \\
& & $g_i$ & \affexpr{127}{63} & \affexpr{157}{366} & \affexpr{277}{138} & \affexpr{163}{129} & \affexpr{13}{294} & \affexpr{229}{114} \\
\midrule
\multirow{2}{*}{$C_7$}
& \multirow{2}{*}{384}
& $f_i$ & \affexpr{1}{66} & \affexpr{169}{357} & \affexpr{121}{183} & \affexpr{49}{54} & \affexpr{241}{102} & \affexpr{121}{63} \\
& & $g_i$ & \affexpr{65}{296} & \affexpr{1}{96} & \affexpr{257}{0} & \affexpr{305}{70} & \affexpr{193}{120} & \affexpr{65}{216} \\
\midrule
\multirow{2}{*}{$C_8$}
& \multirow{2}{*}{576}
& $f_i$ & \affexpr{1}{90} & \affexpr{289}{408} & \affexpr{325}{441} & \affexpr{1}{168} & \affexpr{397}{195} & \affexpr{433}{540} \\
& & $g_i$ & \affexpr{1}{48} & \affexpr{289}{392} & \affexpr{97}{72} & \affexpr{409}{134} & \affexpr{97}{328} & \affexpr{385}{272} \\
\midrule
\multirow{2}{*}{$C_9$}
& \multirow{2}{*}{768}
& $f_i$ & \affexpr{763}{435} & \affexpr{679}{69}  & \affexpr{397}{330} & \affexpr{61}{18}   & \affexpr{697}{612} & \affexpr{373}{246} \\
& & $g_i$ & \affexpr{289}{496} & \affexpr{257}{640} & \affexpr{625}{200} & \affexpr{41}{524}  & \affexpr{193}{672} & \affexpr{449}{672} \\
\midrule
\multirow{2}{*}{$C_{10}$}
& \multirow{2}{*}{768}
& $f_i$ & \affexpr{557}{626} & \affexpr{161}{624} & \affexpr{385}{704} & \affexpr{737}{592} & \affexpr{721}{760} & \affexpr{49}{264}  \\
& & $g_i$ & \affexpr{571}{639} & \affexpr{55}{681}  & \affexpr{229}{294} & \affexpr{307}{579} & \affexpr{637}{234} & \affexpr{121}{660} \\
\bottomrule
\end{tabular}
\end{table}

In the examples below, support sets are sometimes written with two-dimensional
coordinates.  The convention is as follows.  If a vector in $\Ftwo^{LB}$ is
viewed as $L$ consecutive blocks of length $B$, then the pair $(c,t)$ denotes
the one-dimensional coordinate $cB+t$.  Thus $B=P$ for vectors in the original
length $LP$, while $B=Q$ for quotient vectors of length $LQ$.

\begin{example}[Checking the $X$-side entry $\dlat_X=48$ for $C_{10}$]
The $X$-side entry $\dlat_X=48$ for $C_{10}$ is \markededit{obtained from} a single latent
vector.  The computation is carried out in the latent row space: first compute
the kernel of the latent compatibility map, then inspect its low-weight
block-periodic elements, and finally map each surviving latent vector through
$\Htil_X$ and apply the row-space exclusion test.  More explicitly, in the
present $(J,L)=(3,12)$ situation the latent condition is diagonal, so the
search may be carried out one latent row block at a time.  We look for sparse
vectors $u\in\Ker(\Psi_3^{\mathsf T})$ and then evaluate the actual full-length
weight of $\vlambda^{\mathsf T}\Htil_X$.  The search tests periodic supports
\[
u_t=e_t+e_{t+192}+e_{t+384}+e_{t+576}\qquad(0\le t<192),
\]
by checking whether $\Psi_3^{\mathsf T}u_t=0$ and then computing the full-length
image weight after placing $u_t$ in each latent row block.  In this search, the following kernel
vector is found.  Let $e_t$ be the standard basis vector of $\Ftwo^{768}$, put
\[
u=e_0+e_{192}+e_{384}+e_{576},
\]
and place this vector in the first latent row block:
\[
\vlambda=(u,0,0)\in(\Ftwo^{768})^3.
\]
For the affine data of $C_{10}$, one checks that
$\Psi_3^{\mathsf T}u=0$, hence
$\vx:=\vlambda^{\mathsf T}\Htil_X$ satisfies $H_Z\vx^{\mathsf T}=0$.
The support of this vector is
\[
\operatorname{supp}(\vx)
=
\{(c,\alpha_c+192q):0\le c\le 11,\ 0\le q\le 3\},
\]
where
\[
(\alpha_0,\ldots,\alpha_{11})
=(16,184,72,50,48,128,3,42,84,63,105,102).
\]
Thus $\wt(\vx)=12\cdot 4=48$.  Finally, row reduction over $\Ftwo$ gives
\[
\rank(H_X)=2302,\qquad
\rank\!\begin{bmatrix}H_X\\ \vx\end{bmatrix}=2303,
\]
so $\vx\notin\Row(H_X)$.  Proposition~\ref{prop:latent-upper-en} therefore
\markededit{gives} the latent upper bound $d_X^{(\mathrm{lat})}\le 48$, which is the
reported $X$-side \dlat entry for $C_{10}$.
The same number appears on the $Z$ side in Table~\ref{tab:p768-representative-en}
because, for the standard $(J,L)=(3,12)$ APM design used here, active
orthogonality kills all residual blocks except the single interaction
$\Psi_3$.  The $X$-side latent condition uses $\Psi_3^{\mathsf T}$, while the
$Z$-side latent condition uses $\Psi_3$.  Thus the two latent searches are
transpose versions of the same residual-interaction test.  This symmetry is a
feature of this APM family and active/latent split, not a general property of
arbitrary CSS codes.
\end{example}

\begin{example}[Checking the $X$-side full-fiber restricted-lift value \texorpdfstring{\(\dblk_X=32\)}{dbar-blk-X=32} for \(C_1\)]
For $C_9$, take the block factor $m=4$ and write $Q=768/4=192$.
The search is done after imposing the $4$-block-constant ansatz.  This replaces
the original parity-check kernel equation by the smaller quotient equation
$\bar H_Z^{(4)}\bar{\vx}^{\mathsf T}=0$ over length $12Q$.  One then searches
the quotient kernel for a low-weight vector whose lift is outside both
$\Row(H_X)$ and the latent row space.  Concretely, this low-weight search is
done in the shorter quotient coordinate set: for a trial support
$I\subseteq\{0,\ldots,12Q-1\}$, solve
\[
\bar H_Z^{(4)}[I]\bar{\vy}^{\mathsf T}=0
\qquad\text{over }\Ftwo,
\]
where $\bar H_Z^{(4)}[I]$ denotes the submatrix with columns in $I$.  A nonzero
solution gives a quotient vector supported in $I$, and the block-constant lift
then repeats each nonzero quotient coordinate in four positions of the original
length.  On the quotient side there is an
$X$-side vector
$\bar{\vx}\in\Ker(\bar H_Z^{(4)})$ with
\[
\operatorname{supp}(\bar{\vx})
=
\{(0,97),(0,117),(1,145),(2,29),(4,77),(4,141),(5,5),(5,125)\}.
\]
Its block-constant lift has support
\[
\operatorname{supp}(\vx)
=
\{(c,t+192q):(c,t)\in\operatorname{supp}(\bar{\vx}),\ 0\le q\le 3\},
\]
and therefore $\wt(\vx)=4\cdot 8=32$.
The row-space tests give
\[
\vx\notin\Row(H_X),
\qquad
\vx\notin\Row(\Htil_X).
\]
Proposition~\ref{prop:nonlatent-upper-en} therefore gives $d_X\le 32$,
showing the full-fiber block-compression special case of the restricted-lift
method.  This is the $X$-side method-specific full-fiber value $\dblk_X$ for
$C_9$.  Since the selected-fiber value below is smaller, Table~\ref{tab:p768-representative-en}
keeps the restricted-lift components separated.
\end{example}

\begin{example}[Checking the $X$-side selected-fiber restricted-lift value \texorpdfstring{\(\dfib_X=24\)}{dbar-fib-X=24} for \(C_9\)]
For the same code $C_9$, use the fiber-quotient lift with
\[
m=4,\qquad Q=192,\qquad S=\{0,2\}.
\]
The search method is the same quotient-kernel search as above, but the lift is
restricted to the two selected fibers in $S$ rather than to all four fibers.
Equivalently, one forms the restricted fiber-quotient check matrix and
searches its kernel for a low-weight quotient vector whose lifted vector is not
in $\Row(H_X)$.  If $T_S:\Ftwo^{12Q}\to\Ftwo^{12P}$ denotes this restricted
lift, the equation being solved is
\[
H_Z\,T_S(\bar{\vx})^{\mathsf T}=0.
\]
The low-weight search is again performed on quotient supports: choose a trial
support $I$ in the $12Q$ quotient coordinates, solve the corresponding
restricted linear system for $(H_Z\circ T_S)[I]$, lift the resulting quotient
vector only along the fibers in $S$, and then apply the row-space test.
The restricted quotient problem contains an $X$-side vector
$\bar{\vx}$ of weight $12$ in the relevant quotient kernel.  Lifting it only
along the two fibers in $S$ gives a vector $\vx$ satisfying
$H_Z\vx^{\mathsf T}=0$ and $\vx\notin\Row(H_X)$, with
\[
\wt(\vx)=|S|\wt(\bar{\vx})=2\cdot 12=24.
\]
Thus Proposition~\ref{prop:fiber-quotient-upper-en} gives $d_X\le 24$.
Among the displayed $X$-side restricted-lift entries for $C_9$, \markededit{this}
weight is the smallest, so it is recorded in the $\dfib_X$ column.
\end{example}

\begin{example}[Checking the $X$-side CRT-stripe restricted-lift value \texorpdfstring{\(\dcrt_X=96\)}{dbar-crt-X=96} for \(C_9\)]
The CRT-compression special case for $C_9$ uses the factorization
\[
768=3\cdot 256.
\]
Here $\mathcal{S}^{\mathrm{crt}}_{3,256}$ is the full-length CRT stripe subspace
from Definition~\ref{def:crt-stripe-subspace-en}.  Since this code has
12 column blocks, a vector in this subspace is parametrized, block by block, by
coefficients
\[
\alpha_{c,r}\quad (0\le c<12,\ 0\le r<3),
\qquad
\beta_{c,s}\quad (0\le c<12,\ 0\le s<256).
\]
For the coordinate $(c,t)$, this parametrization means
\[
\vx_{cP+t}
=
\sum_{r=0}^{2}\alpha_{c,r}\mathbf{1}_{t\equiv r\!\!\!\pmod 3}
+
\sum_{s=0}^{255}\beta_{c,s}\mathbf{1}_{t\equiv s\!\!\!\pmod {256}}
\quad\text{in }\Ftwo.
\]
Thus choosing the coefficients is exactly the same as choosing a vector
$\vx\in\mathcal{S}^{\mathrm{crt}}_{3,256}$.
If $T:\Ftwo^{12(3+256)}\to\Ftwo^{12P}$ denotes this linear parametrization,
then restricting the parity-check kernel equation to the CRT stripe subspace
means solving
\[
H_Z\,T(\vy)^{\mathsf T}=0
\]
over $\Ftwo$, or equivalently finding the kernel of the restricted check map
$H_Z\circ T$.  After computing a basis of this restricted kernel, the
low-weight part is found in coefficient space: sparse linear combinations of
kernel basis vectors are formed, mapped by $T$, and scored by their full
Hamming weight in length $12P$.  This final scoring is necessary because
different CRT stripes may overlap and cancel over $\Ftwo$.  The surviving
low-weight images are then tested against $\Row(H_X)$.  In this way one obtains
an $X$-side vector
\[
\vx\in \mathcal{S}^{\mathrm{crt}}_{3,256}\cap\Ker(H_Z)
\]
with $\wt(\vx)=96$ and $\vx\notin\Row(H_X)$.  With the coordinate convention
above, the first few support coordinates are
\[
(1,13),(1,29),(1,45),\ldots,
\]
with the remaining coordinates obtained from the same CRT stripe pattern.
Proposition~\ref{prop:crt-compression-upper-en} therefore \markededit{gives}
$d_X\le 96$.  This verifies the $X$-side CRT-stripe entry $\dcrt_X=96$ for
$C_9$.
\end{example}

\begin{example}[Reading the restricted-lift entries]
Table~\ref{tab:p768-representative-en} does not collapse the restricted-lift
value into a single \drls entry.  For example, on the $X$ side of $C_9$ one has
\[
\dblk_X=32,\qquad \dfib_X=24,\qquad \dcrt_X=96.
\]
These are three separate displayed \markededit{entries}, not one optimized
restricted-lift parameter.  In this row, the selected-fiber \markededit{entry} has the
smallest displayed weight among the three restricted-lift entries.  The $Z$ side
is recorded analogously through
$\dblk_Z$, $\dfib_Z$, and $\dcrt_Z$ whenever witnesses are available.
\end{example}

\begin{example}[Checking the $X$-side entry $\ddir_X=74$ for $C_{10}$]
For $C_{10}$, the direct CSS-search method does not prescribe a quotient,
fiber, or CRT subspace.  It searches directly in the CSS kernel
$\Ker(H_Z)$, using the row-space test only as the final exclusion of stabilizer
vectors.  Concretely, the candidate set is generated by repeated
restricted-support kernel searches.  For each trial support
$I\subseteq\{0,\ldots,12P-1\}$, let $H_Z[I]$ be the submatrix of $H_Z$
consisting of the columns indexed by $I$.  One solves
\[
H_Z[I]\vy^{\mathsf T}=0
\qquad\text{over }\Ftwo.
\]
A nonzero solution $\vy$ is then extended by zero outside $I$, giving
$\vx=\iota_I(\vy)\in\Ker(H_Z)$.  The search repeats this step for many
trial supports, enumerates low-weight vectors in the restricted kernels, and
keeps only those passing the final test $\vx\notin\Row(H_X)$.  Since the trial
supports are only a search device, not part of the \markededit{proof}, the final
reported vector is always rechecked against the full matrix $H_Z$ and the full
row space $\Row(H_X)$.
This search finds an $X$-side vector $\vx$
with
\[
H_Z\vx^{\mathsf T}=0,\qquad
\vx\notin\Row(H_X),\qquad
\wt(\vx)=74.
\]
Equivalently, this is a nontrivial $X$-logical operator of weight $74$ found
without imposing any quotient or compression pattern.  Hence
Proposition~\ref{prop:direct-css-upper-en} gives $d_X\le 74$.
This is the $X$-side \ddir entry for $C_{10}$.
\end{example}

\begin{example}[Checking the $X$-side \texorpdfstring{\(\dets_X=10\)}{dbar-ets-X=10} entry]
\label{ex:ub4-small-en}
The cycle-8 ETS upper bound for $C_1$ is verified by the following support.  Let
$C_1$ be the $P=216$ APM-CSS code listed in Table~\ref{tab:p768-apm-en}.
This code has parameters
\[
n=2592,\qquad \rank(H_X)=646,\qquad \rank(H_Z)=646,\qquad k=1300,
\]
and both Tanner graphs have girth $8$.

The cycle-8 ETS search is performed directly in the Tanner graph of $H_Z$.
Starting from individual 8-cycles, one forms connected unions of such cycles
and keeps the supports whose odd-check boundary is empty.  In this code the
following support is found:
\[
\begin{aligned}
S=\{&(0,0),(1,110),(2,76),(3,72),(4,192),(5,40),\\
    &(7,195),(8,90),(10,123),(11,180)\}.
\end{aligned}
\]
Here the coordinate convention is the same as above, with block length
$B=P=216$.  The support $S$ is a cycle-8-connected ETS of type $(10,0)$; it is
obtained as a connected union of four 8-cycles.  Let $\vx$ be its incidence
vector.  Since the boundary is empty,
\[
H_Z\vx^{\mathsf T}=0,\qquad \wt(\vx)=10.
\]
The final step is the stabilizer exclusion test.  Row reduction over $\Ftwo$
gives
\[
\rank(H_X)=646,\qquad
\rank\!\begin{bmatrix}H_X\\ \vx\end{bmatrix}=647,
\]
so $\vx\notin\Row(H_X)$.  Hence $\vx$ is a nontrivial $X$-logical operator, and
Corollary~\ref{cor:cycle8-css-upper-en} gives
\[
d_X\le 10.
\]
This \markededit{gives} the $X$-side value $\dets_X=10$ for $C_1$.  For the other representative codes in
Table~\ref{tab:p768-representative-en}, the same cycle-8-connected ETS
construction was also checked in the reported search range, but it did not
produce a CSS codeword passing the final row-space exclusion test.  Those rows
therefore have no reported \dets witness in the table.  This entry records only
the \markededit{support above} and does not assert that the same ETS construction
produces codewords for the other displayed codes.
\end{example}

\begin{example}[Checking the $Z$-side entry \texorpdfstring{\(\ddec_Z=10\)}{dbar-dec-Z=10} for \(C_1\)]
For the same code $C_1$, the decoder-failure upper bound can be checked by a
pure residual witness obtained from a failed decoding instance.  The search
procedure is to simulate independent Pauli noise and retain only residuals
satisfying the three mathematical conditions in
Proposition~\ref{prop:decoder-failure-upper-en}: the two errors have the same
syndrome, the residual is pure $X$ or pure $Z$, and the residual is not
stabilizer equivalent to zero.

At physical error rate $p=0.03$, this procedure yields a pure $Z$-type
residual $\vz$ of weight $10$.  The value of \(p\) is only a parameter used to
generate this candidate; the \markededit{check} below uses only the residual itself.
With the coordinate convention used above, its support is
\[
\begin{aligned}
\operatorname{supp}(\vz)=\{&(0,204),(1,26),(2,136),(3,60),(4,180),\\
&(5,100),(7,183),(8,78),(10,111),(11,168)\}.
\end{aligned}
\]
Equivalently, if $\bm{\Delta}$ denotes the full Pauli residual
$\mathbf e+\hat{\mathbf e}$, then its $Z$ component $\vz$ has weight $10$ and
its $X$ component is zero.  Since the two errors have the same syndrome,
$\vz\in\Ker(H_X)$; since the residual is not stabilizer equivalent to zero,
$\vz\notin\Row(H_Z)$.
Proposition~\ref{prop:decoder-failure-upper-en} therefore gives
\[
d_Z\le 10.
\]
This is the $Z$-side \ddec entry for $C_1$.  The entry \markededit{reports} only this
residual: it has zero check syndrome on the relevant side, is pure $Z$-type,
and is outside $\Row(H_Z)$.  No statement is made about how frequently such
residuals occur in decoding simulations.
\end{example}

Table~\ref{tab:p768-representative-en} summarizes the worked examples and the
$P\le 768$ rows displayed in the online table.  The Parameters column
incorporates the smallest displayed upper-bound value for each code, in the
form $[[n,k,d\leq \bar d]]$.  The table separates the $X$ and $Z$ sides and
also separates the restricted-lift components \dblk, \dfib, and \dcrt.
Highlighted entries attain the smallest upper bound reported for that code.
The symbol ``--'' is reserved for the \dets and \ddec columns: it means that no
cycle-8 ETS witness or \markededit{checked decoder residual}, respectively, is
reported for that side.

\begin{table}[t]
\centering
\caption{Examples and online-table rows. Other code lists are available in~\cite{KasaiWeb2026}.}
\label{tab:p768-representative-en}
\scriptsize
\setlength{\tabcolsep}{2pt}
\renewcommand{\arraystretch}{1.08}
\begin{adjustbox}{max width=\textwidth}
\begin{tabular}{cc>{\centering\arraybackslash}p{16.0em}c*{7}{>{\centering\arraybackslash}p{2.2em}}}
\toprule
Code & $P$ & Parameters & side & \dlat & \dblk & \dfib & \dcrt & \ddir & \dets & \ddec \\
\midrule
\multirow{2}{*}{$C_1$} & \multirow{2}{*}{216} & \multirow{2}{=}{$[[n=2592, k=1300, d\leq 10]]$} & $X$ & 24 & 24 & 24 & 54 & 54 & \ubhit{\ensuremath{10^{[4]}}} & -- \\
& & & $Z$ & 24 & 24 & 24 & 64 & 88 & -- & \ubhit{10} \\
\midrule
\multirow{2}{*}{$C_2$} & \multirow{2}{*}{216} & \multirow{2}{=}{$[[n=2592, k=1300, d\leq 14]]$} & $X$ & 24 & 42 & 32 & 54 & 54 & -- & -- \\
& & & $Z$ & 24 & 36 & 18 & 64 & 88 & -- & \ubhit{14} \\
\midrule
\multirow{2}{*}{$C_3$} & \multirow{2}{*}{240} & \multirow{2}{=}{$[[n=2880, k=1444, d\leq 24]]$} & $X$ & \ubhit{24} & 40 & 40 & 40 & 40 & -- & -- \\
& & & $Z$ & \ubhit{24} & 40 & \ubhit{24} & 80 & 92 & -- & \ubhit{24} \\
\midrule
\multirow{2}{*}{$C_4$} & \multirow{2}{*}{264} & \multirow{2}{=}{$[[n=3168, k=1588, d\leq 22]]$} & $X$ & 24 & 44 & 48 & 44 & 44 & -- & -- \\
& & & $Z$ & 24 & 44 & 44 & 44 & 94 & -- & \ubhit{22} \\
\midrule
\multirow{2}{*}{$C_5$} & \multirow{2}{*}{288} & \multirow{2}{=}{$[[n=3456, k=1732, d\leq 24]]$} & $X$ & \ubhit{24} & \ubhit{24} & \ubhit{24} & 64 & 32 & -- & 28 \\
& & & $Z$ & \ubhit{24} & 32 & \ubhit{24} & 64 & 32 & -- & -- \\
\midrule
\multirow{2}{*}{$C_6$} & \multirow{2}{*}{384} & \multirow{2}{=}{$[[n=4608, k=2308, d\leq 24]]$} & $X$ & \ubhit{24} & 64 & 48 & 204 & 128 & -- & 28 \\
& & & $Z$ & \ubhit{24} & 48 & 48 & 204 & 236 & -- & -- \\
\midrule
\multirow{2}{*}{$C_7$} & \multirow{2}{*}{384} & \multirow{2}{=}{$[[n=4608, k=2308, d\leq 48]]$} & $X$ & \ubhit{48} & \ubhit{48} & 64 & 54 & 64 & -- & -- \\
& & & $Z$ & \ubhit{48} & 96 & 64 & 256 & 128 & -- & -- \\
\midrule
\multirow{2}{*}{$C_8$} & \multirow{2}{*}{576} & \multirow{2}{=}{$[[n=6912, k=3460, d\leq 32]]$} & $X$ & 48 & 64 & 64 & 72 & 72 & -- & -- \\
& & & $Z$ & 48 & 64 & \ubhit{32} & 384 & 254 & -- & -- \\
\midrule
\multirow{2}{*}{$C_9$} & \multirow{2}{*}{768} & \multirow{2}{=}{$[[n=9216, k=4612, d\leq 24]]$} & $X$ & 48 & 32 & \ubhit{24} & 96 & 128 & -- & -- \\
& & & $Z$ & 48 & 64 & 64 & 512 & 256 & -- & -- \\
\midrule
\multirow{2}{*}{$C_{10}$} & \multirow{2}{*}{768} & \multirow{2}{=}{$[[n=9216, k=4612, d\leq 48]]$} & $X$ & \ubhit{48} & 108 & 64 & 222 & 74 & -- & -- \\
& & & $Z$ & \ubhit{48} & 128 & 72 & 486 & 256 & -- & -- \\
\bottomrule
\end{tabular}
\end{adjustbox}
\end{table}

The lift sizes discussed here are non-prime-power moduli because of the
commutation constraint, not because of an arbitrary selection rule.  As explained in the companion
commutation-pattern note~\cite{KasaiCommPattern2026}, prime-power moduli are
incompatible with the required affine commutation pattern, whereas CRT-split
moduli realize that pattern.

\section{Discussion and Outlook}
\markededit{This paper organizes upper-bound witnesses, not a general lower bound}
for the full CSS distance.  Appendix~A gives exact certification only for the
latent part.  Existing quantum-LDPC lower-bound
techniques~\cite{TillichZemor2014,PanteleevKalachev2022,LeverrierZemor2022,DinurHsiehLinVidick2023} and classical
minimum-weight algorithms~\cite{Leon1988,Stern1989,CanteautChabaud1998,Grassl2006,LisonekTrummer2016}
do not directly match the active/latent decomposition and stabilizer quotient
used here.

\markededit{The current strategy is therefore to find smaller
upper-bound witnesses whenever possible.}  For each code, several independent
witness classes are checked, and the smallest \markededit{value passing the stated tests} is recorded.  The supplementary website is a live
table of the current best values~\cite{KasaiWeb2026}; it is still being updated
as new searches finish.  This use of a separately updated parameter table is
similar in role to the classical and quantum code tables maintained
in~\cite{GrasslCodeTables}.

\markededit{The data currently available suggest that the best recorded upper bound grows}
roughly linearly with blocklength over the explored range.  This observation
should be read cautiously.  The computational cost also grows with the
blocklength, so larger instances receive fewer effective samples.  For distances
below about \(30\), belief-propagation (BP) decoding experiments
in related finite-length quantum-LDPC studies~\cite{KasaiISIT2025,KasaiDegeneracy2025,KomotoKasaiSharpBP2025} often return
low-weight logical residuals directly.  Beyond that scale, BP residual witnesses
become much harder to obtain.  \markededit{The fact that} the large-scale decoding
experiments have not produced failures corresponding to very small logical
errors \markededit{is only negative evidence, not a proof that the true minimum distance is large}.  In
particular, the true minimum distance of this family could still plateau around
weight \(30\).

\appendix

\section{Exact Latent Lower Bounds via Block-Constant Compression}
\label{app:latent-lower-en}
\dlat becomes exact when block-constant structure gives a matching lower bound.
We first recall the two maps used throughout this appendix:
\[
\pi_{m,L}:U_m(LP)\to\Ftwo^{LQ},\qquad
\iota_{m,L}:\Ftwo^{LQ}\to U_m(LP).
\]
For the principal family $(J,L)=(3,12)$, these maps are the compression and lift
between $U_m(12P)$ and $\Ftwo^{12Q}$.  The theorem below states that, under a
block-constant kernel hypothesis, every latent candidate descends to the
compressed image.

\begin{theorem}[Exact latent lower bounds under block-constant kernel hypotheses]
\label{thm:latent-lower-en}
Let $s:=L/2-J$ and assume $P=mQ$. Suppose
\[
\Ker(H_Z\Htil_X^{\mathsf T})\subseteq U_m(sP),
\qquad
\Ker(H_X\Htil_Z^{\mathsf T})\subseteq U_m(sP).
\]
In other words, every latent coefficient vector lying in the kernel of either
active--latent mixed product is blockwise $m$-block constant. Under these
hypotheses, the proof below shows that the following compressed latent images
are well defined:
\[
\bar{\mathcal L}_X:=\pi_{m,L}(C_Z\cap \Row(\Htil_X)),\qquad
\bar{\mathcal L}_Z:=\pi_{m,L}(C_X\cap \Row(\Htil_Z)).
\]
If every nonzero vector in $\bar{\mathcal L}_X$ has weight at least $\tau_X$, and every nonzero vector in $\bar{\mathcal L}_Z$ has weight at least $\tau_Z$, then
\[
d_X^{(\mathrm{lat})}\ge m\tau_X,\qquad
d_Z^{(\mathrm{lat})}\ge m\tau_Z.
\]
\end{theorem}

\begin{proof}
It suffices to prove the $X$ case. Let $\vx\in C_Z\cap \Row(\Htil_X)$. Then
\[
\vx=\vlambda^{\mathsf T}\Htil_X,\qquad
\vlambda\in\Ftwo^{sP}.
\]
Equivalently, $\vx^{\mathsf T}=\Htil_X^{\mathsf T}\vlambda$.
Since $\vx\in C_Z$,
\[
0=H_Z\vx^{\mathsf T}
=H_Z\Htil_X^{\mathsf T}\vlambda,
\]
so
\[
\vlambda\in\Ker(H_Z\Htil_X^{\mathsf T})\subseteq U_m(sP).
\]
Lemma~\ref{lem:apm-descent-en} therefore implies that $\vx$ is block constant blockwise, so $\vx\in U_m(LP)$ and
\[
\bar{\vx}:=\pi_{m,L}(\vx)\in\bar{\mathcal L}_X
\]
is well defined. Moreover,
\[
\vx=\iota_{m,L}(\bar{\vx}).
\]
If $\vx\neq 0$, then $\bar{\vx}\neq 0$ and
\[
\wt(\vx)=m\,\wt(\bar{\vx})\ge m\tau_X.
\]
Every latent $X$-logical representative is a nonzero vector in $C_Z\cap \Row(\Htil_X)$, so $d_X^{(\mathrm{lat})}\ge m\tau_X$.
The $Z$ case is identical with $\Ker(H_X\Htil_Z^{\mathsf T})\subseteq U_m(sP)$.
\end{proof}

This theorem gives the lower-bound part of exact latent certification.  Once the
kernel structure is known to be block constant, excluding low-weight vectors in
the quotient space proves the stated lower bound in the original code.

To make the kernel hypothesis checkable, one replaces it by a rank computation.

\begin{lemma}[Rank test for block-constant kernels]
\label{lem:kernel-rank-test-en}
Let $\ve_0,\dots,\ve_{P-1}$ be the standard basis of $\Ftwo^P$. Define
\[
B_m\in\Ftwo^{(m-1)Q\times P}
\]
to be the matrix whose rows are
\[
\ve_t+\ve_{t+jQ}
\qquad
(0\le t\le Q-1,\ 1\le j\le m-1).
\]
Then
\[
U_m(P)=\Ker(B_m).
\]
Moreover, for every binary matrix $A\in\Ftwo^{r\times P}$,
\[
\Ker(A)\subseteq U_m(P)
\iff
\Row(B_m)\subseteq \Row(A)
\iff
\rank\!\begin{bmatrix}A\\ B_m\end{bmatrix}=\rank(A).
\]
\end{lemma}

\begin{proof}
The defining equations of $U_m(P)$ are the equalities $x_t=x_{t+jQ}$, equivalently
\[
(\ve_t+\ve_{t+jQ})\vx^{\mathsf T}=0.
\]
Hence $U_m(P)=\Ker(B_m)$. Over $\Ftwo$ one has
\[
(\Ker(A))^\perp=\Row(A),
\]
so
\[
\Ker(A)\subseteq \Ker(B_m)
\iff
\Row(B_m)\subseteq \Row(A).
\]
The row-space inclusion is equivalent to the displayed rank equality.
\end{proof}

The lemma rewrites the block-constant kernel condition as a linear-algebraic
rank condition.  This gives a finite check for the hypothesis used in
Theorem~\ref{thm:latent-lower-en}.

\begin{corollary}[Machine-readable form of the latent kernel hypothesis]
\label{cor:kernel-rank-test-en}
Let $s:=L/2-J$ and define
\[
B_m^{(s)}:=I_s\otimes B_m\in\Ftwo^{s(m-1)Q\times sP}.
\]
Then the kernel hypothesis of Theorem~\ref{thm:latent-lower-en} is equivalent to the two binary rank equalities
\[
\rank\!\begin{bmatrix}H_Z\Htil_X^{\mathsf T}\\ B_m^{(s)}\end{bmatrix}
=\rank(H_Z\Htil_X^{\mathsf T}),
\qquad
\rank\!\begin{bmatrix}H_X\Htil_Z^{\mathsf T}\\ B_m^{(s)}\end{bmatrix}
=\rank(H_X\Htil_Z^{\mathsf T}).
\]
Therefore exact latent certification at block factor $m$ follows from these two
rank checks together with an explicit latent witness and an UNSAT certificate
excluding lower-weight vectors in the compressed latent image.
\end{corollary}

\begin{example}
In the specialized $(J,L)=(3,12)$ family with $\Psi_r=0$ for $r\neq 3$, the mixed products reduce to
\[
H_Z\Htil_X^{\mathsf T}=\diag(\Psi_3^{\mathsf T},\Psi_3^{\mathsf T},\Psi_3^{\mathsf T}),
\qquad
H_X\Htil_Z^{\mathsf T}=\diag(\Psi_3,\Psi_3,\Psi_3).
\]
Thus Theorem~\ref{thm:latent-lower-en} and Corollary~\ref{cor:kernel-rank-test-en}
reduce to checking block-constant structure for the kernels of $\Psi_3$ and
$\Psi_3^{\mathsf T}$.  These are the rank checks used for the exact latent
certificates reported here.
\end{example}

\begin{example}
The exact $P=768$ code printed in Table~1 of~\cite{Kasai2026} also satisfies this block-constant kernel hypothesis. For this code,
\[
\rank(\Psi_3)=576,\qquad \dim\Ker(\Psi_3)=192,
\]
so the relevant block factor is $m=4$ with $Q=192$, and the kernels of both $\Psi_3$ and $\Psi_3^{\mathsf T}$ are $4$-block constant.
The SAT/SMT exclusion on the compressed latent image also succeeds, so for this published code one obtains the exact latent values
\[
d_X^{(\mathrm{lat})}=d_Z^{(\mathrm{lat})}=48.
\]
Hence the full-fiber block-compression improvement from $d\le 48$ to
$d\le 32$ for the code of~\cite{Kasai2026} does not come from sharpening
\dlat; it comes from the full-fiber restricted-lift witness recorded as the
method-specific value \dblk.  This is still only an intermediate upper bound:
for the same code \(C_9\), the proper selected-fiber restricted-lift witness
recorded in Table~\ref{tab:p768-representative-en} gives the stronger bound
\(d\le 24\).
\end{example}

This completes the latent exactness argument. Once block-constant kernel
structure is verified, the latent lower-bound problem is reduced to proving
nonexistence of compressed latent vectors below a target weight threshold, for
example by SAT or SMT.

\bibliographystyle{IEEEtran}
\bibliography{ub_refs}

\end{document}